\documentclass[11pt]{article}

\usepackage[latin1]{inputenc}
\usepackage{amsmath,amsthm}
\usepackage{amsfonts}
\usepackage{amssymb}

\usepackage{graphicx,xy,pstricks}
\input xy
\xyoption{all}

\newcommand{\executeiffilenewer}[3]{%
 \ifnum\pdfstrcmp{\pdffilemoddate{#1}}%
 {\pdffilemoddate{#2}}>0%
 {\immediate\write18{#3}}\fi%
}

\theoremstyle{definition} %%% for statements in roman typeface

 \newtheorem{definition}{Definition}[section]
 \newtheorem{remark}[definition]{Remark}

\theoremstyle{plain}      %%% for statements in italic typeface

 \newtheorem{proposition}[definition]{Proposition}
 \newtheorem{theorem}[definition]{Theorem}
 
 \newtheorem{lemma}[definition]{Lemma}
\newtheorem{conjecture}{Conjecture}
\newtheorem*{theorem*}{Theorem}
\newtheorem{result}{Result}

\DeclareMathOperator{\tr}{Tr}
\DeclareMathOperator{\re}{Re}
\DeclareMathOperator{\im}{Im}

\newcommand{\tw}{T\!w}
\newcommand{\C}{\mathbb{C}}
\newcommand{\la}{\langle}
\newcommand{\ra}{\rangle}
\newcommand{\R}{\mathbb{R}}
\newcommand{\Z}{\mathbb{Z}}
\renewcommand{\P}{\mathbb{P}}
\newcommand{\boM}{\mathcal{M}}
\newcommand{\boT}{\mathcal{T}}
\newcommand{\boL}{\mathcal{L}}

\newcommand{\vp}{\varphi}
\newcommand{\be}{\begin{equation}}
\newcommand{\ee}{\end{equation}}

\newcommand{\su}{\mathrm{SU}_2}
\newcommand{\ba}[1]{\overline{#1}}
\newcommand{\cc}{{\check{c}}}
\newcommand{\br}{\overline{r}}
\renewcommand{\phi}{\varphi}
\newcommand{\ci}{C^{\infty}}
\newcommand{\al}{\alpha}
\renewcommand{\Re}{\operatorname{Re}}
\renewcommand{\Im}{\operatorname{Im}}
\title{Toeplitz operators in TQFT via skein theory}
\date{}
\author{Julien March\'e and Thierry Paul\footnote{Centre de math{\'e}matiques Laurent Schwartz (UMR 7640), Ecole Polytechnique -- 91128 Palaiseau, France}
}

\begin{document}
\maketitle
\begin{abstract}
Topological quantum field theory associates to a punctured surface $\Sigma$, a level $r$ and colors $c$ in $\{1,\ldots,r-1\}$ at the marked points a finite dimensional hermitian space $V_r(\Sigma,c)$. Curves $\gamma$ on $\Sigma$ act as Hermitian operator $T_r^\gamma$ on these spaces. In the case of the punctured torus and the 4 times punctured sphere, we prove that the matrix elements of $T_r^\gamma$ have an asymptotic expansion in powers of $\frac{1}{r}$ and we identify the two first terms using trace functions on representation spaces of the surface in $\su$. We conjecture a formula for the general case. Then we show that the curve operators are Toeplitz operators on the sphere in the sense that $T_r^{\gamma}=\Pi_r f^\gamma_r\Pi_r$ where $\Pi_r$ is the Toeplitz projector and $f^\gamma_r$ is an explicit function on the sphere which is smooth away from the poles. 
Using this formula, we show that under some assumptions on the colors associated to the marked points, the sequence $T^\gamma_r$ is a Toeplitz operator in the usual sense with principal symbol equal to the trace function and with subleading term explicitly computed. 
We use this result and semi-classical analysis in order to compute the asymptotics of matrix elements of the representation of the mapping class group of $\Sigma$ on $V_r(\Sigma,c)$. We recover in this way the result of \cite{tw} on the asymptotics of the quantum 6j-symbols and treat the case of the punctured S-matrix. We conclude with some partial results when $\Sigma$ is a genus 2 surface without marked points.
\end{abstract}
\tableofcontents
\section{Introduction and main results of the paper}\label{intro}
Topological quantum field theory (TQFT) were introduced by E. Witten in 1989 as a physical model for the Jones polynomial of knots (see \cite{witten}). Fix a compact goup $G$, and a representation $V$ of $G$. He defined for any knot $K$ in a 3-manifold $M$ a partition function $Z_r(M,K)$ as a Feynman integral over all connections $A$ on some $G$-bundle over $M$ of the form 
$$Z_r(M,K)=\int \tr_V (\textrm{Hol}_K A) e^{ir\textrm{CS}(A)}\mathrm{d}A.$$ 
In this formula, Hol$_K(A)$ is the holonomy of the connection $A$ along $K$ and CS$(A)$ is the Chern-Simons functional.

The fact that this invariant is indeed computable comes from cut-and-paste rules implied by formal properties of Feynman integration. It allowed E. Witten to recognize (an evaluation of) the Jones polynomial(s) of $K$  and also to predict the asymptotics of such invariants when $r$ goes to infinity. 
If the semi-classical parameter $r$ goes to infinity in the formula for $Z_r(M,K)$, the stationary phase principle implies that the integral should concentrate on critical points of the Chern-Simons functional, that is connections which are flat on $M\setminus K$. Moreover, if $M$ has boundary $\Sigma$, then $Z_r(M,K)$ should be interpreted as an element of the geometric quantization of the moduli space $\boM(\Sigma,G)$ that is gauge equivalence classes of flat $G$-connections on $\Sigma$. 

As far as we know, since then, there is no rigorous geometric construction of this TQFT. One can define the geometric quantization of $\boM(\Sigma,G)$ (although it requires a heavy machinery) but not the state associated to a 3-manifold bounding $\Sigma$. On the other hand, N. Reshetikhin and V. Turaev developed in \cite{rt} a rigorous combinatorial construction of the TQFT. We will use in this article the version of C. Blanchet, G. Masbaum, N. Habegger and P. Vogel (\cite{bhmv}) which works for $G=\su$ and relies only on the combinatorics of the Kauffman bracket. The price to pay for these combinatorial constructions is that the geometry gets hidden, in particular the natural expectations for the semi-classical limit $r\to\infty$ become very mysterious.

We study here curve operators, that is the natural action of curves on $\Sigma$ on the TQFT vector space associated to $\Sigma$. Notice that we use the combinatorial curve operators, as in \cite{mn}, and not the geometric ones, as in \cite{and1,and2,and3} which are by definition Toeplitz operators. In this sense, our strategy differs from Andersen's. Moreover, we do not use the complex structures on $\boM(\Sigma,G)$ parametrized by the Teichm\"{u}ller space of $\Sigma$ but a more explicit one. On the other hand, the moduli space $\boM(\Sigma,G)$ appears only in an indirect way via action-angle coordinates.

\subsection{Main results}\label{main}
Let $\Sigma$ be either a once punctured torus or a 4-times punctured sphere. Then, under generic assumptions on 
the holonomy $t$ around the marked points, the moduli spaces $\boM(\Sigma,t)$ (see Subsection \ref{representation}) are symplectomorphic to the standard sphere $S^2=\C P^1$. Let $V_r(\Sigma,c)$ be the TQFT vector space associated to $\Sigma$ with level $r$ and colors $c$ at marked points (see Subsection \ref{TQFT}).  Any curve $\gamma$ on $\Sigma$ acts on $V_r(\Sigma,c)$ as a Hermitian operator $T_r^{\gamma}$. 
We will define in Subsection \ref{asymptoregime} a sequence of colorings $c_r$ such that $\pi\frac{c_r}r$ converges to $t$ and the dimension of $V_r(\Sigma,c_r)$, denoted by $N$, grows linearly with $r$. One of the main goals of the paper is to realize any curve operator as a Toeplitz operator, once again in the case of the once punctured torus or the 4-times punctured sphere.

\begin{result}[see Theorem \ref{curvetoplitz} in Section \ref{asymptoregime}]
Suppose that $\boM(\Sigma,t)$ is smooth. Then there is a canonical diffeomorphism $\boM(\Sigma,t)\simeq \C P^1$ such that for any curve $\gamma$ on $\Sigma$, the sequence of matrices $(T_r^{\gamma})$ are Toeplitz operators with 
%principal 
symbol $\sigma^\gamma=\sigma_0^\gamma+\frac 1 N\sigma_1^\gamma+O(N^{-2})$ with

\be\label{princ}
\sigma_0^\gamma(\rho)=-\tr\rho(\gamma)
\ee
\begin{equation}\label{sub}
\sigma_1^\gamma=\frac{1}{2}\Delta_S\sigma_0^\gamma
\end{equation}
% \be\label{sub}
%\left[\frac\tau{2(1-\tau)}\partial^2_{\tau^2}+\tau^2\partial_\tau-\frac {\partial^2_{\theta^2}}{8\tau(1-\tau)}
%+\left(\frac {i\partial_{\theta}\partial_\tau}2-\frac{i(2\tau^2+2\tau-1)\partial_\theta}{4\tau(1-\tau)}\right)
%\mathcal {H}_\theta\right]\sigma_0^\gamma
%\ee
where $\Delta_S$ is the Laplacian on the sphere which is equal to $(1+|z|^2)^2\partial_z\partial_{\ba z}$ in the canonical holomorphic coordinate $z$.
%where $\mathcal {H}_\theta$ is the Hilbert transform acting on the  variable $\theta$.
 \end{result}
Let us remark that, though there is no Weyl quantization on the sphere, the condition \eqref{sub} corresponds, on the flat case, to having no Weyl subsymbol.

In fact we prove a somehow more precise result which does not require any smoothness assumption. We prove the existence of an exact (non-semiclassical) symbol:

\begin{result}[see Theorem \ref{thmtoe2} in Section \ref{toe2}]
For a holonomy $t$ around the marked points satisfying some mild assumptions, any level $r\in \mathbb{N}$ big enough and any curve $\gamma\subset\Sigma$, there exist a function $f^r_\gamma$ on the sphere, smooth except possibly at the two poles, 
such that
\be\label{exact}
T_r^{\gamma}=\mathcal T_r^{f^r_\gamma}, 
\ee
where $\mathcal T_r^{f}$ is the Toeplitz quantization, namely $\mathcal T_r^{f}=\Pi_r f \Pi_r$ and $\Pi_r$ is the Toeplitz projector.

Moreover
\be\label{asym}
f^r_\gamma\sim -\tr\rho(\gamma)+
%r^{-1}
%\sigma_1^\gamma
\frac{1}{2N}\Delta_S\sigma_0^\gamma+ \sum_{k=2}^\infty N^{-k}\sigma_k^\gamma,
\ee with 
%$\sigma_1^\gamma=\frac{1}{2}\Delta_S\sigma_0^\gamma$ and  
$\sigma_k^\gamma$ smooth away from the poles and where $N=\dim V_r(\Sigma,c)$. 
\end{result}

Let us remark that this result gives a Toeplitz framework for curve operators even in the singular case, 
namely the case where the trace function is not smooth. In this case, standard asymptotic methods will fail for giving a Toeplitz symbol by iterations.

The proof of this theorem goes through a closed formula relating the exact symbol of a Toeplitz operator
to its matrix elements on the canonical basis. We first prove the following theorem:
\begin{result}[see Theorem \ref{thmprincipal} in Section \ref{struc}]
For any curve $\gamma\subset\Sigma$ there exists a non negative integer $k$ such that if we denote by $F_{n,m}^\gamma$ the matrix elements of $T^\gamma_r$ in the canonical basis of $V_r(\Sigma,c)$ then 
$$F_{n,m}=0\quad\text{ if }|n-m|>k.$$
\end{result}

Let $\tau\in [0,1]$ and $\theta\in \R/2\pi\Z$ be cylindrical coordinates on $S^2\simeq\boM(\Sigma,t)$. 
\begin{result}[see Theorem \ref{thmprincipal} in Section \ref{struc}]
There exist a sequence of $C^\infty$ functions $F(\tau,\theta, r)=\sum_{|\mu|\leq k} F_\mu(\tau,r)e^{i\mu\theta}$ on the sphere 
satisfying for $0<\tau<1$,
\be\label{subbb}
F(\tau,\theta, r)=F(\tau,\theta, \infty)+\frac 1{2ir}\partial_\tau\partial_\theta F(\tau,\theta, \infty)+O(r^{-2})
\ee
with $F(\tau,\theta, \infty)=-\tr\rho(\gamma)$ where $\rho$ is parametrized by $(\tau,\theta)$ and such that the matrix elements $F_{n,m}$ of $T_r^{\gamma}$ 
on the canonical basis are for $0<\frac nr<1$,
\be\label{subbbb}
F_{n, n+\mu}=F_\mu\Big(\frac n r,\frac 1 r\Big).
\ee
\end{result}
We also conjecture in Section \ref{struc} these two last results to be true in the higher genus case.

The curve operators $T_r^{\gamma}$ play a crucial role in understanding the asymptotic properties of TQFT. 
They were the key ingredient for proving the asymptotic faithfulness of the quantum representations, see \cite{and1,fww,mn}. In the case of the torus, they are used extensively in \cite{lj1,lj2}. 
Their spectral decomposition is directly linked to the basis coming from the combinatorial TQFT and their symbol and sub-symbol 
give the semi-classical properties at first order. As an application of Result 1 and 4, we can use standard 
techniques from semi-classical analysis to compute the asymptotics of the pairings between different basis element, 
recovering the quantum 6j-symbols computed in \cite{tw} (see Section \ref{6j}) and as a new example, the $S$-matrix of the punctured torus which is detailed below. Let us remark that the condition \eqref{sub}, which determines the subprincipal contribution of the Toeplitz symbol out of the classical trace function \eqref{princ} enters at leading orders in these asymptotics and 
is somehow crucial in order to get topological invariant terms.

Notice that the case of classical 6j-symbols were treated by a  similar method in \cite{polygon}.

\begin{figure}[htbp]
\centering
  \def\svgwidth{5cm}
 \executeiffilenewer{smat.svg}{smat.pdf}%
 {inkscape -z -D --file=smat.svg %
 --export-pdf=smat.pdf --export-latex}%
 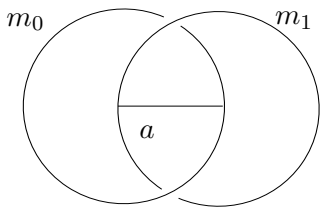%

  \caption{Punctured S-matrix}
  \label{fig:smat} 
\end{figure}

Let $\Gamma$ be the graph shown in Figure \ref{fig:smat} colored by $c=(m_0,m_1,a)$ were $a$ and $D$ are odd integers satisfying $\frac{a}{2}<m_i<D-\frac{a}{2}$ for $i=0,1$. We denote by $\langle \Gamma,c\rangle_r$ the evaluation of the Kauffman bracket of the colored graph $(\Gamma,c)$ at $t=-e^{i\pi/2r}$, see Subsection \ref{jones-wenzl}.
Then we have the following result (see Proposition \ref{prop-smat}):
\begin{result}\label{smatrix}
For any odd $\br$, setting $r=D\br$ one has
$$\langle \Gamma,\br c\rangle_r=\frac{2r}{\pi}N_r\big(G^{-1/4}\cos(\frac{r}{2\pi}S+\frac{\pi}{4})+O(r^{-1})\big)$$
Where 
\begin{itemize}
\item[-] $N_r=\frac{(\langle m_0+\frac{a-1}{2}\rangle!\langle m_0-\frac{a+1}{2}\rangle!\langle m_1+\frac{a-1}{2}\rangle!\langle m_1-\frac{a+1}{2}\rangle!)^{1/2}\langle\frac{a-1}{2}\rangle!^2}{\langle a-1\rangle!\langle m_0-1\rangle!\langle m_1-1\rangle!}$
\item[-] $\langle n\rangle=\sin(\frac{\pi n}{r})$ and $\langle n\rangle!=\prod_{k=1}^n\langle k\rangle$
\item[-] $G=\cos(\frac \alpha 2)^2-\cos(\tau_0)^2-\cos(\tau_1)^2-\cos(\tau_0)^2\cos(\tau_1)^2$
\item[-] $\alpha=\frac{\pi a}{D},\tau_0=\frac{\pi m_0}{D},\tau_1=\frac{\pi m_1}{D}$.
\item[-] $S$ is the area of the moduli space $\{(A,B)\in\su,\tr(ABA^{-1}B^{-1})=2\cos(\alpha),\tr(A)\le 2\cos(\tau_0),\tr(B)\le 2\cos(\tau_1)\}/\sim$
\end{itemize}
\end{result}

\vskip 1cm

Finally we give some partial results concerning the case of a surface of genus 2. 
We identify the TQFT Hilbert space with the geometric quantization of the projective space 
$\P^3$ and 
 consider the three  curves $\gamma,\delta,\eta$ of Figure \ref{fig:genus22}.
\begin{figure}[htbp]
\centering
  \def\svgwidth{6cm}
 \executeiffilenewer{genus2.svg}{genus2.pdf}%
 {inkscape -z -D --file=genus2.svg %
 --export-pdf=genus2.pdf --export-latex}%
 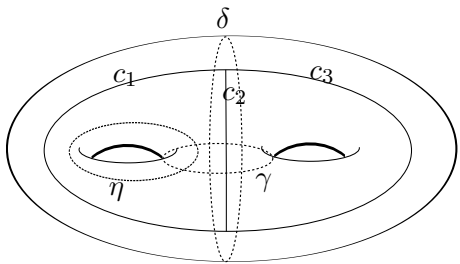%

  \caption{Curves on a genus 2 surface}
  \label{fig:genus22} 
\end{figure}
\begin{result}[see Theorem \ref{symbolgenus} in Section \ref{examples}]
$T^\gamma_r,\ T^\delta_r,\ T^\eta_r$ are Toeplitz operators with exact symbols in the sense of Result 2, 
$f^\gamma_r,\ f^\delta_r,\ f^\eta_r$ 
defined on $\P^3$, admitting a smooth asymptotic expansion in $\P^3/\{z_0z_1z_2z_3=0\}$. 
Moreover $\sigma^\gamma_0=-\tr\rho(\gamma)$ and $\sigma^\delta_0=-\tr\rho(\delta)$. 
\end{result}

We believe that the methods developed in the present paper can  be generalized to the higher genus cases, using the geometric quantization of toric manifolds and the Mellin transform strategy of Section \ref{toe}.
%We therefore make 
 \textbf{Result} 6 suggests, in addition to \textbf{Conjectures} \ref{principal}  and \ref{deformation} below,  the following one :
\begin{conjecture}\label{conj0}
Result 2 is true in the general higher genus case by replacing the sphere by a toric manifold modeling the moduli space $\mathcal M(\Sigma,t)$, the two poles by the singularities of $\mathcal M(\Sigma,t)$ and the Laplacian on the sphere by the one of $\mathcal M(\Sigma,t)$.
\end{conjecture}
\subsection{Organization of the paper}
Section \ref{struc} is devoted to the general structure of curve operators. The preliminary considerations on TQFT exposed in the beginning of the section lead to the definition of trigonometric operators that we conjecture to be the shape of general curve operators, together with the "subprincipal symbol" property \eqref{subbb}. We prove that the conjecture is true in the case of the punctured torus and the $4$-times punctured sphere in Section \ref{small} by first computing the matrix of three particular curve operators. By relying heavily on the properties of the Kauffman algebra, we then show that the latter computation is sufficient to prove the statement for any curve in these two punctured surfaces.

Section \ref{toe} is somehow the heart of the paper and contains the results on the Toeplitz structure of curve operators. After deriving an expression for the total symbol of a Toeplitz operator out of its matrix elements on a canonical basis, we show that the machinery applies to the case of any curve operator, even in the singular cases. We also compute the asymptotic regime and show that the total symbol is a classical one in the regular case. 

Section \ref{mapping} is devoted to the computation of pairing formulas, namely basis change matrices. We obtain a general asymptotic formula for the matrix elements of the quantum representations of the mapping class group of the two punctured surfaces. In particular we recover, by pure semiclassical methods, the asymptotic of $6j$-symbols and derive the case of the punctured $S$-matrix.
In Section \ref{genus2} we show how our methods apply in the case of genus 2 surface through the quantization of the projective space $\P^3$.

The first author was supported by the Agence Nationale de la Recherche ANR-08-JCJC-0114-01
and the second  is a member of the CNRS.

\section{Structure of curve operators in TQFT}\label{struc}

\subsection{Basics of TQFT}\label{TQFT}

Let $\Sigma$ be a closed oriented surface with marked points $p_1,\ldots,p_n$. 
Fix an integer $r>0$ and let $\mathcal{C}_r=\{1,\ldots,r-1\}$ be the set of colors. Given a coloring $c=(c_1,\ldots,c_n)\in\mathcal{C}_r^n$ of the marked points, the construction of \cite{bhmv} gives a finite dimensional hermitian vector space $V_r(\Sigma,c)$. In their notation, $p=2r$, $A=-\exp(i\pi/2r)$ (notice that we shifted by 1 the colors).

One can construct a basis of this space by considering a banded graph $\Gamma$ that is a finite graph with $n$ univalent vertices labeled by $p_1,\ldots,p_n$ and trivalent vertices with local orientations around vertices. This datum allows to construct a surface $S$ which retracts on $\Gamma$ and such that the univalent vertices belong to the boundary of $S$. 

Let $H$ be the 3-manifold $S\times[0,1]$. A \emph{presentation} of $\Sigma$ is an homeomorphism $h:\partial H\to \Sigma$ which respects orientation and such that for all $i$, $h(p_i\times\{1/2\})=p_i$, see Figure \ref{graphe-surface}. 

Let $E$ be the set of edges of $\Gamma$. We call admissible coloring of $\Gamma$ a map $\cc:E\to \mathcal{C}_r$ such that the following conditions hold:
\begin{itemize}
\item[-] for each edge $e$ connected to a univalent vertex $p_i$ one has $\cc_e=c_i$.
\item[-] for any triple of edges $e,f,g$ adjacent to the same vertex one has 
\item[-] \center{$\cc_e+\cc_f+\cc_g$ is odd}
\item[-] $\cc_e+\cc_f< \cc_g$
\item[-] $\cc_e+\cc_f+\cc_g<2r$.
\end{itemize}

The construction of \cite{bhmv} provides for each admissible coloring $\cc$ a vector $\phi_{\cc}\in V_r(\Sigma,c)$ obtained by cabling the graph $\Gamma$ by a specific combination of multicurves. Moreover, the family $(\phi_{\cc})$ when $\cc$ runs over all admissible colorings is a Hermitian basis of $V_r(\Sigma,c)$. This construction will be sketched in Subsection \ref{jones-wenzl}.

\subsection{Representation spaces in SU$_2$}\label{representation}
Fix as before a surface $\Sigma$ with marked points $p_1,\ldots,p_n$. Suppose that these points are colored by $(t_1,\ldots,t_n)\in(\mathbb{Q}\pi)^n$ and denote by $\gamma_i$ a curve going around $p_i$. We define the following moduli space:
$$\mathcal{M}(\Sigma,t)=\{\rho:\pi_1(\Sigma\setminus\{p_1,\ldots,p_n\})\to\su\text{ s.t. }\forall i,\tr\rho(\gamma_i)=2\cos(t_i)\}/\sim.$$ One has $\rho\sim \rho'$ if there is $g\in$ SU$_2$ such that $\rho'=g\rho g^{-1}$. 

This space is a compact symplectic variety. It is smooth for generic values of $t$.

Let $U$ be the set of all maps $\tau:E\to [0,\pi]$ such that for any edge $e$ adjacent to a marked point $p_i$ one has $\tau_i=t_i$ and  for any triple of edges $e,f,g$ adjacent to the same vertex one has 

\begin{itemize}
\item[-] \center{$\tau_e+\tau_f\le \tau_g$}
\item[-] $\tau_e+\tau_f+\tau_g\le 2\pi$.
\end{itemize}

Given a curve $\gamma$ in $\Sigma$, we define a function $h_{\gamma}:\mathcal{M}(\Sigma,t)\to [0,\pi]$ by the formula $h_{\gamma}(\rho)=\textrm{acos}(\frac{1}{2}\tr\rho(\gamma))$.
Let $p:\mathcal{M}(\Sigma,t)\to U$ be the map defined by $p(\rho)_e=h_{C_e}(\rho)$ where $C_e$ is the circle dual to $e$ in $\Sigma$. The map $p$ is a continuous surjective map which is a smooth Lagrangian fibration over the interior of $U$. 

Moreover, there is a preferred section $s:U\to \mathcal{M}(\Sigma,t)$ defined in the following way: for any edge $e$ we define a circle $D_e$ distinct from $C_e$: if the edge $e$ matches two distinct vertices in $\Gamma$, then $D_e$ goes along the edge $e$, cut $C_e$ in two points and no other circle $C_f$. If the edge $e$ join a vertex to itself, then $D_e$ is the curve lying at the boundary of $S$ going along $e$. It cuts $C_e$ once and no other curve $C_f$. Some examples are shown in Figure \ref{graphe-surface}.

\begin{figure}
\centering
  \def\svgwidth{\columnwidth}
 \executeiffilenewer{graphe-surface.svg}{graphe-surface.pdf}%
 {inkscape -z -D --file=graphe-surface.svg %
 --export-pdf=graphe-surface.pdf --export-latex}%
 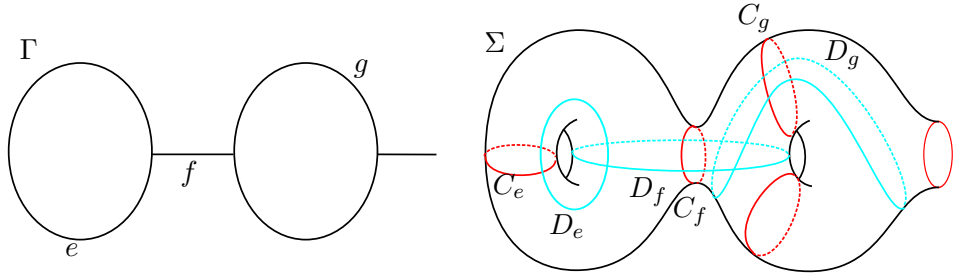%

  \caption{From a graph to a surface with a pants decomposition}
  \label{graphe-surface}
\end{figure}

For $\theta_e$ in $S^1=\R/2\pi\Z$, we denote by $\theta_e.\rho$ the action on $\rho$ of the Hamiltonian flow of the map $h_{C_e}(\rho)=p(\rho)_e$. 
\begin{lemma}\label{angle-coor}
For any $\tau$ in the interior of $U$, there is a unique $\rho\in \mathcal{M}(\Sigma,t)$ which minimizes simultaneously the functions $h_{D_e}$ on the fiber of $p$.
\end{lemma}
\begin{proof}
By Goldman's formula, the Poisson bracket of two functions $h_{\delta}$ and $h_{\delta}$ vanishes if $\gamma$ and $\delta$ do not intersect (see \cite{goldman}). As  $D_e$ intersects $C_f$ if and only if $e=f$, we see that all minimizations can be done independently on each edge. The proof follows by inspection of the two cases, that is the 4 times punctured sphere and the once-punctured torus which is done in Lemmas \ref{toregeom} and \ref{spheregeom}.
\end{proof}

\subsection{The Kauffman algebra}
For any oriented 3-manifold $M$, let $K(M,A)$ be the quotient of the free $\C[A^{\pm 1}]$-module generated by isotopy classes of banded links in $M$ modulo the Kauffman relations shown in Figure \ref{kauf}.

\begin{figure}[width=8cm,height=3cm]
\begin{center}
\begin{pspicture}(-2,0)(3,3)
\includegraphics{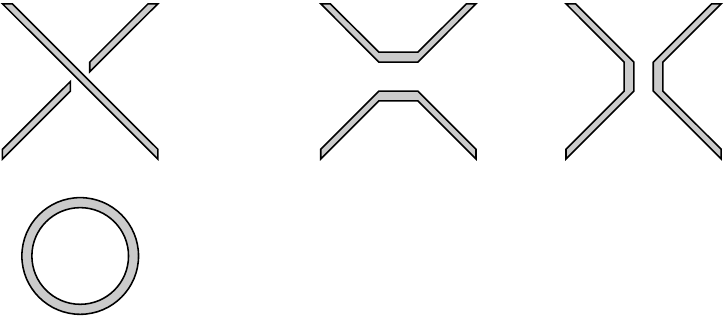}
\put(-5.2,2.3){$=A$}
\rput(-1.9,2.5){$+A^{-1}$}
\rput(-4,0.7){$=(-A^2-A^{-2})\quad \emptyset$}
\end{pspicture}
\caption{Kauffman relations}
\label{kauf}
\end{center}
\end{figure}
Write $K(\Sigma,A)=K((\Sigma\setminus\{p_1,\ldots,p_n\})\times [0,1],A)$.
This module is an algebra where the product $\delta\cdot \gamma$ of two banded links is obtained by stretching $\delta$ and $\gamma$ so that they leave respectively in $\Sigma\times[1/2,1]$ and $\Sigma\times[0,1/2]$ and then glueing the two intervals.

We call \emph{multicurve} on $\Sigma$ any 1-submanifold avoiding the marked points such that no component bounds a disc avoiding the marked points. It is well-known that $K(\Sigma,A)$ is a free $\C[A^{\pm 1}]$-module generated by multicurves. Using this preferred basis, we identify $K(\Sigma,\zeta)=K(\Sigma,A)\otimes \C_{A=\zeta}$ with $K(\Sigma,-1)$ for any $\zeta\ne 0$ and we embed $K(\Sigma,-\exp(i\pi\hbar/2))=K(\Sigma,A)\otimes \C[[h]]_{A=-\exp(i\pi\hbar/2)}$ in $K(\Sigma,-1)[[\hbar]]$.

We have the following theorem:
\begin{theorem}
The map $K(\Sigma,-1)\to C^{\infty}(\mathcal{M}(\Sigma,t))$ defined by $\gamma\mapsto f_\gamma$ where $f_{\gamma}(\rho)=-\tr \rho(\gamma)$ is an injective morphism of algebras.

For any two multicurves $\gamma,\delta$ considered as elements of $K(\Sigma,-1)[[\hbar]]$, one has 
$$\gamma\cdot\delta=f_{\gamma} f_{\delta}+\frac{\hbar}{2i}\{f_{\gamma},f_{\delta}\}+o(\hbar).$$
For each $r>0$, the Kauffman algebra at $\zeta_r=-e^{i\pi/2r}$ acts on $V_r(\Sigma,c)$ where a curve $\gamma$ acts by a Hermitian operator $T_r^{\gamma}$ called curve operator.
\end{theorem}

\subsection{The curve operators in TQFT}
In this section, we give some details on the skein construction of TQFT in order to explain the definition and computation of the curve operators. 

Let $\Sigma$ be a closed surface and $H$ be a handlebody such that $\partial H=\Sigma$.
Fix an integer $r$ and consider the Kauffman module $K(H,\zeta_r)$. For any embedding $j$ of $H$ in $S^3$, we define a subspace 
$$N^j_r=\{x\in  K(H,\zeta_r), \forall y\in K(S^3\setminus j(H),\zeta_r), \quad \langle x,y \rangle=0 \in K(S^3,\zeta_r)\}$$
Suppose that a 3-manifold $M$ is the union of two 3-manifolds $M_1\cup M_2$. There is a natural pairing $K(M_1,A)\times K(M_2,A)\to K(M,A)$ defined by sending two banded links $L_1\subset M_1$, $L_2\subset M_2$ to their union $L_1\amalg L_2$ viewed as a banded link in $M_1\cup M_2$. 
This is the meaning of the pairing in the formula above for $M_1=H,M_2=S^3\setminus j(H)$ and $M=S^3$.

The point is that $N^j_r$ is a subspace of finite codimension in $K(H,\zeta_r)$ which does not depend on $j$. We set then $V_r(\Sigma)=K(H,\zeta_r)/N^j_r$.

Given a embedding $j:H\to S^3$ it defines an embedding of $\Sigma=\partial H$. Thicken slightly its image in $S^3$ such that one has the decomposition
$$H\cup \Sigma\times[0,1] \cup S^3\setminus j(H)=S^3$$
We deduce from it a natural map 
$$K(H,\zeta_r)\times K(\Sigma,\zeta_r) \times K(S^3\setminus j(H),\zeta_r)\to K(S^3,\zeta_r)=\C$$ denoted by $(x,y,z)\mapsto\langle x|y|z\rangle$. If for any $z$ we have $\langle x|\emptyset|z\rangle=0$ then for any $y$ we also have $\langle x|y|z\rangle=0$ by "pushing $y$ to the right". 
It follows that "pushing $y$ to the left" defines an action of $K(\Sigma,\zeta_r)$ on $V_r(\Sigma)$. 
 
\subsection{Jones-Wenzl idempotents, colorings and marked points}\label{jones-wenzl}
In order to deal with the case when $\Gamma$ has boundary points and to construct the TQFT basis mentioned in Subsection \ref{TQFT}, we need to introduce the Temperley-Lieb algebras $\boT_l$ and the Jones-Wenzl idempotents $f_l\in\boT_l$. Let $P_l\subset (0,1)$ be a finite set with $l$ elements and $\boT_l(A)$ be the $\C[A^{\pm 1}]$-module generated by banded tangles $L\subset [0,1]^3$ such that $\partial L=P_l\times\{1/2\}\times\{\pm 1\}$ modulo Kauffman relations. The product $L_1\cdot L_2$ of two tangles is given by stacking $L_1$ above $L_2$. 

One can define recursively the Jones-Wenzl idempotents by $f_0=0,f_1=1$ and  the relations of Figure \ref{fig:jones-wenzl} where we set $[l]=\frac{A^{2l}-A^{-2l}}{A^2-A^{-2}}$.
\begin{figure}
\centering
  \def\svgwidth{7cm}
 \executeiffilenewer{jones-wenzl.svg}{jones-wenzl.pdf}%
 {inkscape -z -D --file=jones-wenzl.svg %
 --export-pdf=jones-wenzl.pdf --export-latex}%
 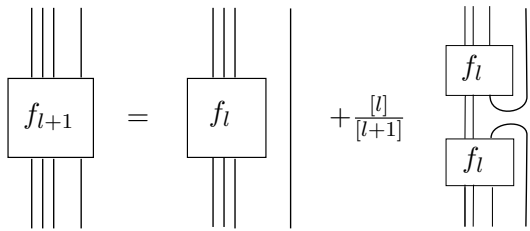%

  \caption{Recurrence relation for the Jones-Wenzl idempotents}
  \label{fig:jones-wenzl} 
\end{figure}
If $A=\zeta_r$, the idempotents $f_l$ are well-defined in $\boT_r$ provided that we have $l<r$. 

Suppose now that $\Sigma$ is a surface with punctures $p_1,\ldots,p_n$. Fix a level $r$ and chose colors $c=(c_1,\ldots,c_n)\in \mathcal{C}_r^n$. We define $V_r(\Sigma,c)$ in the following way. 
Consider for all $i\in \{1,\ldots,n\}$ a subset $P_i\in\Sigma$ of cardinality $c_i-1$ and lying in a small neighborhood of $p_i$. We define the relative skein module $K(H,c,A)$ as the $\C[A^{\pm 1}]$-module generated by banded tangles in $H$ whose intersection with $\Sigma$ is $\bigcup_i P_i$ modulo the Kauffman relations. 
Then for any embedding $j:H\to S^3$ we set 
$$N^j_r=\{x \in K(H,c,\zeta_r),\quad\forall z\in K(S^3\setminus j(H),c,\zeta_r)\quad \langle x| \bigotimes_{i=1}^{r-1} f_{c_i-1}|z\rangle=0 \} $$
We set as before $V_r(\Sigma,c)=K(H,c,\zeta_r)/N^j_r$ and the curve operator $T^\gamma_r$ is defined in the same way provided that the curve $\gamma$ does not touch the punctures $p_1,\ldots,p_n$.

We are ready to define the basis of the latter space. Let $\Gamma$ be a trivalent graph encoding a pants decomposition of $\Sigma$ as in Subsection \ref{TQFT}. Let $S$ be the surface containing $\Gamma$ and set $H=S\times [0,1]$. 
Given any admissible map $\cc:E\to\mathcal{C}_r$ we define $\psi_\cc\in K(H,c,\zeta_r)$ in the following way: 
\begin{itemize}
\item[-] Replace each edge $e$ of $\Gamma$ by $\cc_e-1$ parallel copies lying on $S$. 
\item[-] Insert in the middle of each edge the idempotent $(-1)^{\cc_e-1}f_{\cc_e-1}$
\item[-] In the neighborhood of each trivalent vertex, join the three bunches of lines in $S$ in the unique possible way avoiding crossings. 
\end{itemize}
It happens that this family of vectors is an orthogonal basis of $V_r(\Sigma,c)$ for a natural Hermitian structure we do not define here. We refer to Theorem 4.11 in \cite{bhmv} for the proof and the following formula:
\be\label{norm}
||\psi_\cc||^2=\Big(\frac{2}{r}\Big)^{\chi(\Gamma)/2}\frac{\prod_v \langle c_v^1,c_v^2,c_v^3\ra }{\prod_e \la c_e\ra}
\ee
where for any trivalent vertex $v$ of $\Gamma$ denote by $c_v^1,c_v^2,c_v^3$ the colors of the edges incoming at $v$ and for any internal edge $e$, $c_e$ denotes the color of that edge. We also set $\la n\ra=\sin(\pi n/r)$, $\la n\ra!=\prod_{k=1}^n\la k\ra$ and 
\be
\la a,b,c\ra=\frac{\la i+j+k+1\ra!\la i\ra!\la j\ra!\la k \ra!}{\la j+k\ra!\la i+k\ra!\la i+j\ra!}
\ee
for $i,j,k$ defined by $a=j+k+1,b=i+k+1$ and $c=i+j+1$
The vectors $\phi_\cc$ used in this article are given by the formula $\phi_\cc=\frac{\psi_\cc}{||\psi_\cc||}$. 

\subsection{Trigonometric operators}
Let $t=(t_1,\ldots t_n)$ be a family in $([0,\pi]\cap \mathbb{Q}\pi)^n$. Let $D$ be the common denominator of the rational numbers $t_j/\pi$. We will denote by $c_r$ the sequence of colorings $c_r=(r\frac{t_j}{\pi})\in \mathcal{C}_r^n$ and suppose implicitly that $r$ is a multiple of $D$ so that this coloring take integral values.
\begin{definition}
A family of operators $T_r\in \textrm{End}(V_r(\Sigma,c_r))$ is called \emph{trigonometric} if there is an open subset $V\subset U\times [0,1]$ containing $\text{Int}(U)\times \{0\}$,  a finite family of smooth functions $F_k:V\to \mathbb{R}$ indexed by maps $k:E\to \mathbb{Z}$ such that for all admissible colorings $\cc$, one has $T_r \phi_{\cc}=\sum_k F_k(\frac{\pi\cc}{r},\frac{\pi}{r}) \phi_{\cc+k}$.
\end{definition}
 Any multicurve $\gamma$ gives rise in TQFT to a family of operators $T^{\gamma}_r\in \textrm{End}(V_r(\Sigma,c_r))$: these operators are called curve operators. We make the following conjecture:
\begin{conjecture}\label{principal}
For any multicurve $\gamma$, the curve operator $T^{\gamma}_r$ is trigonometric.
The coefficients $F_k$ are recursively computable and verify the following properties:
\begin{enumerate}
\item
$F_k$ vanishes if for some edge $e$ the geometric intersection of $\gamma$ with $C_e$ is lower than $k_e$. 
\item
The map $F_k(\cdot,0):U\to\mathbb{R}$ is the $k$-th Fourier coefficient of the function $f_{\gamma}$ with respect to the action of $(S^1)^E$ decribed in \cite{goldman,multicurve}.
\end{enumerate}
\end{conjecture}
We define the $\psi$-symbol of $T^{\gamma}_r$ as the expression $\sigma^{\gamma}=\sum_{k}F_k(\tau,\hbar)e^{ik\theta}$ where $\tau$ is an element of the interior of $U$, $\hbar$ is a real parameter, $\theta\in (\R/2\pi\Z)^E,k\in \Z^E$ and $k\theta=\sum_e k_e\theta_e$. We remark that the data $(\tau,\theta)$ are action-angle coordinates on the open subset of $\boM(\Sigma,t)$ defined as $p^{-1}(\textrm{Int}(U))$. Hence, $\psi$-symbol $\sigma^{\gamma}$ may be interpreted as a deformation of the trace function $f_{\gamma}$. The terminology $\psi$-symbol stands for pseudo-differential operator.

The following conjecture gives the first order of the deformation:

\begin{conjecture}\label{deformation}
For any multicurve $\gamma$, the $\psi$-symbol $\sigma^{\gamma}$ has the following asymptotic development: 
$$\sigma^{\gamma}=f_{\gamma}+\left(\frac{1}{2i}\sum_e\frac{\partial^2f_{\gamma}}{\partial \theta_e\partial \tau_e}\right)\hbar+o(\hbar)$$
\end{conjecture}

The statements of Conjecture \ref{principal} easily transfer to statements on the $\psi$-symbol $\sigma^{\gamma}$, hence the proof should consist in computing this symbol as precisely as possible. 
This conjecture is by no means inaccessible: we prove it for the case of the once-punctured torus and the 4th-punctured sphere by analyzing three particular cases and using the structure of the Kauffman algebra on both surfaces (linked to SL(2,$\Z$) and the Farey triangulation). We believe that a general proof will use a detailed study of the fractional Dehn twists on multicurves.
\begin{theorem}\label{thmprincipal}
The conjectures \ref{principal} and \ref{deformation} hold if $\Sigma$ is a punctured torus or a 4 times punctured sphere.
\end{theorem}

\section{Small genus cases}\label{small}
The proof of Theorem \ref{thmprincipal} will proceed from an explicit computation in two particular cases and various compatibilities with the Kauffman algebra. 
The easiest case is the punctured torus.
\subsection{The once-punctured torus}\label{punctured-torus}
Take $\Sigma$ a punctured torus, and $\Gamma$ the graph with one trivalent vertex and one univalent vertex as shown in Figure \ref{fig:torep}. We call $a\in \mathcal{C}_r$ the (odd) color of the marked point and $n$ the color of the loop $e$. Let $\gamma$ be the circle around the loop ($\gamma=C_e$ in the notation of the preceding section) and $\delta$ the circle parallel to the loop ($\delta=D_e$).

\begin{figure}[htbp]
\begin{center}
\begin{pspicture}(0,0)(1,4.5)
\rput[b1](0,0){\includegraphics{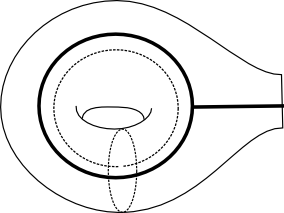}}
\rput[b](2.4,1.8){$\Gamma$}
\rput[b](-2.5,0.5){$\Sigma$}
\rput[b](-0.9,0.2){$\gamma$}
\rput[b](-0.9,3){$\delta$}
\rput[b](3.2,2.2){$a$}
\rput[b](-2.4,2.2){$n$}
\end{pspicture}
\caption{Basis for the punctured torus}
\label{fig:torep}
\end{center}
\end{figure}
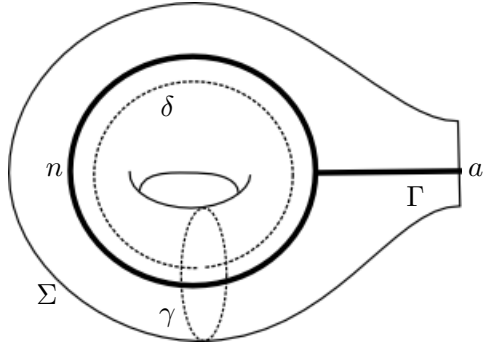

\begin{proposition}\label{torep}
For $n$ satisfying $a/2< n< r-a/2$ one has:
\begin{eqnarray*}
T_{r}^{\gamma} \phi_n&=&-2\cos(\frac{\pi n}{r})\phi_n\\
T_r^{\delta}\phi_n&=&u_{n+1}\phi_{n+1}+u_n\phi_{n-1}\text{ where}\\ 
u_n&=&-\left(\frac{\la n+\frac{a-1}{2}\ra\la n-\frac{a+1}{2}\ra}{\la n\ra\la n-1\ra}\right)^{1/2}
\end{eqnarray*}
\end{proposition}
\begin{proof}
These formulas come from standard computations using fusion rules: the preceding section and Appendix \ref{fusion} contain all the information needed to do the computation. \end{proof}
Let $\tau=\frac{\pi n}{r},\alpha=\frac{\pi a}{r}$ and $\hbar=\frac{\pi}{r}$. Then, we compute directly from the above formula the $\psi$-symbols of $\gamma$ and $\delta$ namely $\sigma^{\gamma}=-2\cos(\tau)$ and 
\begin{eqnarray*}
\sigma^{\delta}=&-\big(\frac{\sin(\tau+\alpha/2+\hbar/2)\sin(\tau-\alpha/2+\hbar/2)}{\sin(\tau)\sin(\tau+\hbar)}\big)^{1/2}e^{i\theta}\\
&-\big(\frac{\sin(\tau+\alpha/2-\hbar/2)\sin(\tau-\alpha/2-\hbar/2)}{\sin(\tau)\sin(\tau-\hbar)}\big)^{1/2}e^{-i\theta}
\end{eqnarray*}

The first point of Conjecture \ref{principal} is obviously satisfied as $\gamma$ and $\delta$ intersect $\gamma$ respectively 0 and 1 times. The second point is a consequence of the following lemma:

\begin{lemma}\label{toregeom}
Let $\Sigma$ be a punctured torus, $\alpha\in [0,\pi]$, $\gamma$ and $\delta$ two curves on $\Sigma$ intersecting once. Let $(\tau,\theta)$ be the action-angle coordinates given in Lemma \ref{angle-coor} and such that $\tau=h_{\gamma}$. Then, $f_{\delta}=-2\frac{\sqrt{\sin(\tau+\alpha/2)\sin(\tau-\alpha/2)}}{\sin(\tau)}\cos(\theta)$.
\end{lemma}
\begin{proof}
Choose a base-point $x$ on the boundary of $\Sigma$ and represent $\gamma$ and $\delta$ as elements of $\pi_1(\Sigma,x)$. For any $\rho:\pi_1(\Sigma,x)\to\su$, set $A=\rho(\gamma)$ and $B=\rho(\delta)$. The condition for $\rho$ to be in $\boM(\Sigma,\alpha)$ is $\tr(ABA^{-1}B^{-1})=2\cos(\alpha)$. For any $\theta\in \R/2\pi\Z$, set $U_{\theta}=\begin{pmatrix} e^{i\theta}&0\\ 0&e^{-i\theta}\end{pmatrix}$. As $\tr(A)=2\cos(\tau)$, one can suppose up to conjugation that we have $A=U_{\tau}$. The Hamiltonian flow of $h_{\gamma}$ acts on $(A,B)$ by the formula $\theta.A=A$ and $\theta.B=U_{\theta}B$, see \cite{goldman}. 
Write $B=\begin{pmatrix} b_1&b_2\\-\ba{b_2}&\ba{b_1}\end{pmatrix}$ such that $|b_1|^2+|b_2|^2=1$. The formula $\tr(ABA^{-1}B^{-1})=2\cos(\alpha)$ implies that $|b_1|^2=\frac{\cos(\alpha)-\cos(2\tau)}{1-\cos(2\tau)}=\frac{\sin(\tau+\alpha/2)\sin(\tau-\alpha/2)}{\sin(\tau)^2}$.

We also compute $\tr(\theta.B)=2\re(b_1 e^{i\theta})=2\frac{\sqrt{\sin(\tau+\alpha/2)\sin(\tau-\alpha/2)}}{\sin(\tau)}\cos(\theta+\psi)$ where $\psi=\arg(b_1)$. Recall that in Lemma \ref{angle-coor}, the origin for the angle $\theta$ is such that $\tr(\theta.B)$ is maximal. Hence, one can suppose that $\psi=0$ which proves the lemma. \end{proof}

The geometric picture can be nicely visualized in the coordinates $x=\tr(A), y=\tr(B), z=\tr(AB)$. The moduli space $\boM(\Sigma,\alpha)$ is isomorphic to the cubic $\{(x,y,z)\in [-2,2]^3, x^2+y^2+z^2-xyz-2=2\cos(\alpha)\}$. The level sets of $x$ are ellipses on which the Hamiltonian flow $\theta$ acts by rotation. The level set $\theta=0$ is given by the maximum of $y$ on the level sets of $x$. It is then defined as the intersection of the cubic with the half plane $\{z=0,y\ge 0\}$.

Let $\tw^\gamma$ be the Dehn twist along $\gamma$ and denote by $\xi$ the curve $\tw^\gamma(\delta)$. 
\begin{lemma}\label{derive-tore}
For $\eta\in\{\gamma,\delta,\xi\}$, the following equation holds:

$$\sigma^{\eta}=f_{\eta}+\left(\frac{1}{2i}\frac{\partial^2f_{\eta}}{\partial \theta\partial \tau}\right)\hbar+o(\hbar)$$
\end{lemma}
\begin{proof}
This is trivial for $\eta=\gamma$ as one has $\sigma^\gamma=f_\gamma$ and the first order term vanishes. 
For any $\eta$, set $\sigma^\eta=\sum_k F^\eta_k(\tau,\hbar)e^{ik\theta}$. The lemma is an immediate consequence of the following identity:
$$\partial_\hbar F_k^{\eta}=\frac{k}{2}\partial_\tau F_k^\eta$$
We check this identity by interpreting it has $\partial_\hbar F_k^\eta(\tau-\frac{k\hbar}{2},\hbar)$. 
We compute $F_1^\delta(\tau-\hbar/2,\hbar)=-\big(\frac{\sin(\tau+\alpha/2)\sin(\tau-\alpha/2)}{\sin(\tau-\hbar/2)\sin(\tau+\hbar/2)}\big)^{1/2}=F_1(\tau,0)+o(\hbar).$ The same argument works for $F_{-1}^\delta$.

Denote by $\tw_r^\gamma$ the action of $\tw^\gamma$ on $V_r(\Sigma,a)$.
Following \cite{bhmv} p.913,  the twist acts on $\phi_n$ by the formula $\tw_r^\gamma\phi_n=e^{\frac{i\pi}{2r}(n^2-1)}\phi_n$. Hence, from the identity $T_r^{\xi}=\tw_r^\gamma T_r^{\delta}(\tw_r^\gamma)^{-1}$, we obtain $T_r^{\xi}\phi_n=u_{n+1}e^{\frac{i\pi}{2r}(2n+1)}\phi_{n+1}+u_n e^{\frac{i\pi}{2r}(-2n+1)}\phi_{n-1}$ and $$\sigma^{\xi}=e^{i\hbar/2}\sigma^{\delta}(\tau,\theta+\tau).$$

When computing $\partial_{\hbar}\sigma^{\xi}$ and $\frac{1}{2i}\partial_{\tau\theta}f_{\xi}$,  we obtain the same result as before changing $\theta$ into $\theta+\tau$ plus a term $\frac{i\pi}{2}$. This extra term cancel with the term produced in the derivation, proving the formula for $\xi$.
\end{proof}
\subsection{The 4th-punctured sphere}
Consider a 4 times punctured sphere $\Sigma$, and let $\Gamma$ be the graph of Figure \ref{fig:spherep} with two univalent vertices and 4 univalent vertices. We call $a,b,c,d$ the colors of the marked points and $n$ the color of the internal edge $e$. Let $\zeta=C_e$ and $\eta=D_e$.

\begin{figure}[htbp]
\begin{center}
\begin{pspicture}(0,0)(1,4.5)
\rput[b1](0,0){\includegraphics{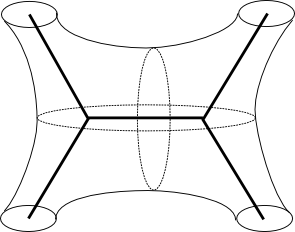}}
\rput[b](2.2,1.2){$\Gamma$}
\rput[b](-3,1){$\Sigma$}
\rput[b](-0.9,1.8){$\eta$}
\rput[b](-0.3,3.4){$\zeta$}
\rput[b](-2.7,4.5){$a$}
\rput[b](2.7,4.5){$b$}
\rput[b](2.7,0.2){$c$}
\rput[b](-2.7,0.2){$d$}
\rput[b](0.1,2.18){$n$}
\end{pspicture}
\caption{Basis for the punctured sphere}
\label{fig:spherep}
\end{center}
\end{figure}
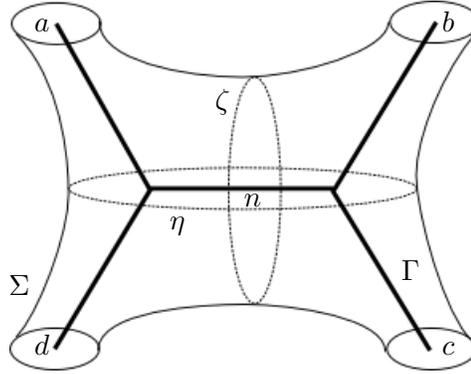

\begin{proposition}\label{spherep}
For $n$ satisfying $\max(|a-d|,|b-c|)< n<\min(a+d,2r-a-d,b+c,2r-b-c)$ and $n+1=a+d=b+c\mod 2$, one has:
\begin{eqnarray*}
T_{r}^{\zeta} \phi_n&=&-2\cos(\frac{\pi n}{r})\phi_n\\
T_r^{\eta}\phi_n&=& w_{n-2}\phi_{n-2}+v_n\phi_n+w_n\phi_{n+2}\text{ where}\\
w_n&=&-4\big( \frac{\langle \frac{a+d-n-1}{2}\rangle\la \frac{a-d+n+1}{2}\ra\la\frac{-a+d+n-1}{2}\ra\la \frac{a+d+n+1}{2}\ra}{\la n\ra\la n+1\ra}\cdot\\
&&\frac{\la \frac{b+c-n-1}{2}\ra\la \frac{b-c+n+1}{2}\ra\la \frac{-b+c+n+1}{2}\ra\la \frac{b+c+n+1}{2}\ra}{\la n+1\ra\la n+2\ra}\big)^{1/2}\\
v_n&=&-2\cos(\frac{\pi(c+d-1)}{r})\\
&&-4\frac{\la \frac{a+d-n-1}{2}\ra\la \frac{a-d+n+1}{2}\ra\la \frac{b+c-n-1}{2}\ra\la \frac{b-c+n+1}{2}\ra}{\la n\ra\la n+1\ra}\\
&& -4\frac{\la \frac{a+d+n-1}{2}\ra\la \frac{-a+d+n-1}{2}\ra\la \frac{b+c+n-1}{2}\ra\la \frac{-b+c+n-1}{2}\ra}{\la n-1\ra\la n\ra}
\end{eqnarray*}
\end{proposition}
\begin{proof}
Again, we omit the proof since it is a long and standard computation using fusion rules, see Appendix \ref{fusion}.
\end{proof}

Let $\tau=\frac{\pi n}{r},\hbar=\frac{\pi}{r},\alpha=\frac{\pi a}{r},\beta=\frac{\pi b}{r},\gamma=\frac{\pi c}{r},\delta=\frac{\pi d}{r}$. Then, the $\psi$-symbols of $T^{\zeta}$ and $T^{\eta}$ are respectively $\sigma^{\gamma}=-2\cos(\tau)$ and 
$$\sigma^{\eta}= -I(\tau,\hbar)-J(\tau,\hbar)e^{2i\theta}-J(\tau-2\hbar,\hbar)e^{-2i\theta} \textrm{ where }$$
\begin{eqnarray*}
I(\tau,\hbar)&=& 2\cos(\gamma+\delta-\hbar)\\
 &&+4
\frac{\sin(\frac{\alpha+\delta-\tau-\hbar}{2})\sin(\frac{\alpha-\delta+\tau+\hbar}{2})\sin(\frac{\beta+\gamma-\tau-\hbar}{2})\sin(\frac{\beta-\gamma+\tau+\hbar}{2})}
{\sin(\tau)\sin(\tau+\hbar)}\\
&&+4
\frac{\sin(\frac{\alpha+\delta+\tau-\hbar}{2})\sin(\frac{-\alpha+\delta+\tau-\hbar}{2})\sin(\frac{\beta+\gamma+\tau-\hbar}{2})\sin(\frac{-\beta+\gamma+\tau-\hbar}{2})}
{\sin(\tau)\sin(\tau-\hbar)}
\end{eqnarray*}
and
\begin{eqnarray*}
J(\tau,\hbar)&=&4\Big(\frac{
\sin(\frac{\alpha+\delta-\tau-\hbar}{2})\sin(\frac{\alpha-\delta+\tau+\hbar}{2})\sin(\frac{-\alpha+\delta+\tau+\hbar}{2})\sin(\frac{\alpha+\delta+\tau+\hbar}{2})
}{\sin(\tau)\sin(\tau+\hbar)}\\
&&\frac{\sin(\frac{\beta+\gamma-\tau-\hbar}{2})\sin(\frac{\beta-\gamma+\tau+\hbar}{2})\sin(\frac{-\beta+\gamma+\tau+\hbar}{2})\sin(\frac{\beta+\gamma+\tau+\hbar}{2})}
{\sin(\tau+\hbar)\sin(\tau+2\hbar)}
\Big)^{1/2}
\end{eqnarray*}
As for the case of the punctured torus, the first point of Conjecture \ref{principal} is trivially satisfied. We interpret the formulas where $\hbar=0$ in the following lemma.
\begin{lemma}\label{spheregeom}
Fix $\alpha,\beta,\gamma,\delta\in [0,\pi]$ and let $\boM=\{(A,B,C,D)\in \su, \textrm{ s. t. }\\ABCD=1, \tr(A)=2\cos(\alpha),\tr(B)=2\cos(\beta),\tr(C)=2\cos(\gamma),\tr(D)=2\cos(\delta)\}/\sim$. Let $(\tau,\theta)$ be the action-angles coordinates on $\boM$ given by Lemma \ref{angle-coor} such that $\tau=h_{\zeta}$. Then, $f_{\eta}=\sigma^{\eta}|_{\hbar=0}$.
\end{lemma}
This lemma uses the following easy sub-lemma which we do not prove.
\begin{lemma}
Let $A,B\in\su$ be two matrices such that $AB=U_{\tau}$ for some $\tau\in \R\setminus \pi\Z$. Then, the upper left entry of $A$ is $\frac{\cos(\alpha)\sin(\tau)+i(\cos(\beta)-\cos(\alpha)\cos(\tau))}{\sin(\tau)}$ where $\tr(A)=2\cos(\alpha)$ and $\tr(B)=2\cos(\beta)$.
\end{lemma}
\begin{proof}
One has $\tr(AD)=2\cos(\tau)$ and $\tr(AB)=2\cos(h_{\eta})$. One can suppose up to conjugation that $DA=U_{\tau}$. Then, the Hamiltonian flow of $h_{\zeta}$ acts on the 4-tuple $(A,B,C,D)$ by $\theta.(A,B,C,D)=(A,U_{\theta}BU_{-\theta},U_{\theta}CU_{-\theta},D)$. 
If one writes $A=\begin{pmatrix} a_1&a_2\\-\ba{a_2}&\ba{a_1}\end{pmatrix}$ and $B=\begin{pmatrix} b_1&b_2\\-\ba{b_2}&\ba{b_1}\end{pmatrix}$ such that $|a_1|^2+|a_2|^2=|b_1|^2+|b_2|^2=1$, we compute 
$$\tr((\theta.A)(\theta.B))=2\re(a_1b_1)-2\re(a_2\ba{b_2}e^{-2i\theta})$$

As $A^{-1}D^{-1}=U_{-\tau}$ and $BC=U_{-\tau}$, we find using the sub-lemma that $\ba{a_1}=\frac{\cos(\alpha)\sin(\tau)+i(\cos(\alpha)\cos(\tau)-\cos(\delta))}{\sin(\tau)}$ and $b_1=\frac{\cos(\beta)\sin(\tau)+i(\cos(\beta)\cos(\tau)-\cos(\gamma))}{\sin(\tau)}$.
We obtain by tedious trigonometric computations:
\begin{eqnarray*}
2\re(a_1b_1)&=&\!\!\mbox{\scriptsize $\frac{2}{\sin(\tau)^2}\big(\cos(\alpha)\cos(\beta)\!+\!\cos(\gamma)\cos(\delta)\!-\!\cos(\tau)(\cos(\alpha)\cos(\gamma)+\cos(\beta)\cos(\gamma))\big)$}\\
&=&I(\tau,0)
\end{eqnarray*}

We compute in the same way
\begin{eqnarray*}
|a_2|^2&=&1-|a_1|^2=\frac{\sin(\tau)^2-\cos(\alpha)^2-\cos(\delta)^2+2\cos(\alpha)\cos(\delta)\cos(\tau)}{\sin(\tau)^2}\\
&=&\frac{-4\sin(\frac{\alpha+\delta-\tau}{2})\sin(\frac{\alpha-\delta+\tau}{2})\sin(\frac{-\alpha+\delta+\tau}{2})\sin(\frac{\alpha+\delta+\tau}{2})}{\sin(\tau)^2}
\end{eqnarray*} 
Using this relation and the same one involving $|b_2|^2$, we find $|a_2b_2|=J(\tau,0)$.
Hence, we have $-f_{\eta}=I(\tau,0)+2J(\tau,0)\cos(2\theta+\psi)$.
By the same argument as in Lemma \ref{toregeom}, we now that $\theta$ is chosen such that $-f_{\eta}$ is maximal. Hence, one can suppose that $\psi=0$ which proves the lemma. 
\end{proof}

Again, there is a nice geometric picture of the preceding lemma. Consider the map $\boM\to \su^4/\su$ which sends the 4-tuple $(A,B,C,D)$ to the 4-tuple $P_1=1,P_2=A,P_3=AB,P_4=ABC$. Identifying $\su$ to the sphere $S^3$, the length between two elements $P$ and $Q$ is $\textrm{acos}(\frac{1}{2}\tr(PQ^{-1}))$. Hence, the length $l_{ij}$ between $P_i$ and $P_j$ satisfy $l_{12}=\alpha,l_{23}=\beta,l_{34}=\gamma,l_{14}=\delta, l_{13}=\tau$. We may then interpret $\boM$ as the moduli space of spherical quadrilaterals with fixed lengths. The coordinate $\tau$ is the length of the diagonal $P_1P_3$ whereas the angle $2\theta$ is the dihedral angle at the edge $P_1P_3$. The coordinate $h_{\eta}$ is equal to the length $l_{24}$. For a fixed value of $\tau$, it is minimal for non-convex planar quadrilaterals. Hence, the level set $\theta=0$ consists in non-convex planar quadrilaterals whereas the level set $\theta=\pi/2$ consists in convex planar quadrilaterals. Notice that contrary to the case of the punctured torus, the angle coordinate is $\pi$-periodic: this is explained by the fact that the curve $\zeta$ is separating, see \cite{goldman}.

Let $\xi$ be the curve obtained by performing on $\eta$ a half Dehn twist along $\zeta$ as shown in Figure \ref{demi-twist}.
\begin{lemma}\label{derive-sphere}
For $\lambda\in\{\zeta,\eta,\xi\}$, one has 
$$\sigma^\lambda=f_{\lambda}+\left(\frac{1}{2i}\frac{\partial^2f_\lambda}{\partial\theta\partial\tau}\right)\hbar+o(\hbar)$$
\end{lemma}
\begin{proof}
As for Lemma \ref{derive-tore}, the formula is trivial for the case $\lambda=\zeta$ and for the case $\lambda=\eta$, it will be a consequence of the equation 
$$F_k^\eta(\tau-k\hbar/2,\hbar)=F_k^\eta(\tau,0)+o(\hbar)$$

We check directly that one has $J(\tau-\hbar,\hbar)=o(\hbar)$ which implies that the above formula works for $k=2$ and $k=-2$. For $k=0$, one has to show 
$I(\tau,\hbar)=I(\tau,0)+o(\hbar)$. This was checked by a computation with Sage.
The computation for $\xi$ is done in the same spirit as in Lemma \ref{derive-tore} whereas instead of using the action of the twist, we use the half-twist coefficient $H(c;a,b)=(-1)^\epsilon \exp\big(\frac{i\pi}{4r}(c^2-a^2-b^2+1)\big)$ where $\epsilon=(-a^2-b^2-c^2+2ab+2bc+2ac-2a-2b-2c+3)/4$ (see \cite{mv}).
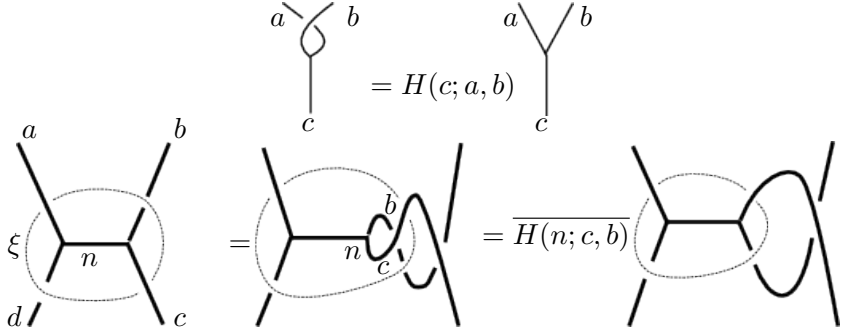
\begin{figure}[htbp]
\begin{pspicture}(0,0)(1,4.5)
\rput[b1](6,2.8){\includegraphics{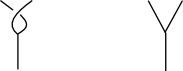}}
\rput[b](6,0){\includegraphics{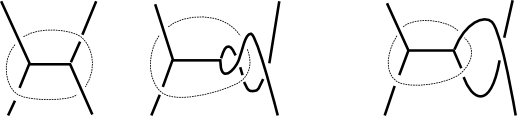}}
\rput[b](6.2,3){$=H(c;a,b)$}
\rput[b](4.4,2.6){$c$}
\rput[b](5,4){$b$}
\rput[b](4,4){$a$}
\rput[b](7.5,2.6){$c$}
\rput[b](8.1,4){$b$}
\rput[b](7.1,4){$a$}
\rput[b](5,0.9){$n$}
\rput[b](5.4,0.7){$c$}
\rput[b](5.5,1.5){$b$}
\rput[b](3.5,1){$=$}
\rput[b](7.7,1){$=\overline{H(n;c,b)}$}
\rput[b](0.5,0.9){$\xi$}
\rput[b](0.7,2.5){$a$}
\rput[b](2.7,2.5){$b$}
\rput[b](2.7,0){$c$}
\rput[b](0.5,0){$d$}
\rput[b](1.5,0.8){$n$}
\end{pspicture}
\caption{Half-twist acting on curves}
\label{demi-twist}
\end{figure}
\end{proof}
In Figure \ref{demi-twist}, we observe that we can compute the operator $T_r^\xi$ knowing $T_r^\eta$ and formulas for the half twist. We find precisely 
\begin{eqnarray*}
T_r^\xi \phi_n&=&\!\!\frac{H(n-2;b,c)}{H(n;c,b)}w_{n-2}\phi_{n-2}\!+\!\frac{H(n;b,c)}{H(n;c,b)}v_n \phi_n\!+\!\frac{H(n+2;b,c)}{H(n;c,b)}w_n\phi_{n+2}\\
&=&-e^{i\pi(n-1)/r}w_{n-2}\phi_{n-2}+v_n \phi_n-e^{-i\pi(n+1)/r}w_n\phi_{n+2}
\end{eqnarray*}
Here, we used the formula for $w_n,v_n$ associated to the operator $T^\eta_r$ given in Proposition \ref{spherep} where we implicitly inverted $b$ and $c$. From the last equation, we compute the symbol $\sigma^\xi=-e^{i(\tau-\hbar)}F^{\eta}_{-2}e^{-2i\theta}+F^{\eta}_0-e^{-i(\tau+\hbar)}F^{\eta}_2$. The equation of the lemma is satisfied by a direct computation.

\subsection{The general case}
Consider a surface $\Sigma$ with marked points $p_1,\ldots,p_n$ colored by rational multiples of $\pi$ denoted by  $t_1,\ldots,t_n$. Let $\Gamma$ be a graph associated to a pants decomposition of $\Sigma$. 
\begin{definition}
We say that a multicurve $\gamma$ satisfies the property $(*)$ if
\begin{enumerate}
\item
$T^{\gamma}_{r}$ is a trigonometric operator (denote by $\sigma^{\gamma}$ its $\psi$-symbol).
\item
$\sigma^{\gamma}=f_{\gamma}+\left(\frac{1}{2i}\sum_e\frac{\partial^2f_{\gamma}}{\partial \theta_e\partial \tau_e}\right)\hbar+o(\hbar)$
\end{enumerate}
\end{definition}
We extend this definition to $K(\Sigma,-1)[[\hbar]]$ by $\C[[\hbar]]$-linearity.

\begin{proposition}\label{produit}
Let $\gamma$ and $\delta$ be two multicurves satisfying (*) then their product in $K(\Sigma,-e^{i\pi\hbar/2})$ satisfies also (*).
\end{proposition}
\begin{proof}
Let $\sigma^{\gamma}=\sum_k F_k(\tau,\hbar)e^{ik\theta}$ and $\sigma^{\gamma}=\sum_l F_l(\tau,\hbar)e^{il\theta}$ be the $\psi$-symbols of $T^{\gamma}_r$ and $T^{\delta}_r$. A direct computation shows that the product of the two operators is trigonometric with the following $\psi$-symbol:
\begin{equation}\label{starproduit}
\sigma^{\gamma}\star\sigma^{\delta}=\sum_n e^{in\theta}\sum_{k+l=n}F_k(\tau,\hbar)G_l(\tau+k\hbar,\hbar).
\end{equation}
We compute the product $\gamma\cdot \delta$ in the Kauffman algebra by stacking the two multicurves and smoothing all crossings. As the map $\gamma\mapsto T^{\gamma}_r$ is a morphism from $K(\Sigma,\zeta_r)$ to $\text{End}(V_r(\Sigma,rt/\pi))$, we obtain the first point of the proposition.

Letting $\hbar$ go to $0$ in the formula \eqref{starproduit}, we get the following expansion:
$$\sigma^{\gamma}\star\sigma^{\delta}=\sigma^{\gamma}\sigma^{\delta}+\frac{\hbar}{i}\partial_{\theta}\sigma^{\gamma}\partial_{\tau}\sigma^{\delta}+o(\hbar).$$
On the other hand, as an element of $K(\Sigma,-1)[[\hbar]]$, we have $\gamma\cdot\delta=f_{\gamma}f_{\delta}+\frac{\hbar}{2i}\{f_\gamma,f_\delta\}$, see \cite{turaev}. Hence, we see that the $\psi$-symbols coincide at 0-th order on $\hbar$. To equal the first order terms, one need to show the following: 
\begin{align*}
f_{\gamma}(\frac{1}{2i}\partial_\theta\partial_\tau f_{\delta})+f_{\delta}(\frac{1}{2i}\partial_\theta\partial_\tau f_{\gamma})+\frac{1}{i}\partial_{\theta}f_\gamma\partial_{\tau}f_{\delta}
=\frac{1}{2i}\partial_\theta\partial\tau(f_{\gamma}f_{\delta})+\frac{1}{2i}\{f_{\gamma},f_{\delta}\}.
\end{align*}
This formula holds thanks to Leibniz formula and the fact that $(\theta,\tau)$ are action angle coordinates, meaning that $\{f_\gamma,f_\delta\}=\partial_{\theta}f_\gamma\partial_\tau f_\delta-\partial_{\tau}f_\gamma\partial_\theta f_\delta$.
\end{proof}

We are now able to prove Theorem \ref{principal} by using the following results on the Kauffman algebra of the 1-torus and 4-sphere which are taken from \cite{bp}.

Consider the Farey triangulation of the hyperbolic disc $\mathbb{H}$. We identify the boundary of $\mathbb{H}$ with $\mathbb{P}^1(\mathbb{R})$ and consider the set of rational points $\mathbb{P}^1(\mathbb{Q})$ in the boundary. A point $[a,b]$ correspond to the simple closed curve in the punctured torus with slope $a/b$ (unique up to isotopy). The number of intersection points between two curves $[a_1,b_1]$ and $[a_2,b_2]$ is $|a_1b_2-a_2b_1|$. Join two points by a geodesic in $\mathbb{H}$ if and only if the corresponding curves intersect in one point. We get the Farey triangulation as shown in Figure \ref{fig:farey}.

\begin{figure}[htbp]
\centering
  \def\svgwidth{6cm}
 \executeiffilenewer{farey.svg}{farey.pdf}%
 {inkscape -z -D --file=farey.svg %
 --export-pdf=farey.pdf --export-latex}%
 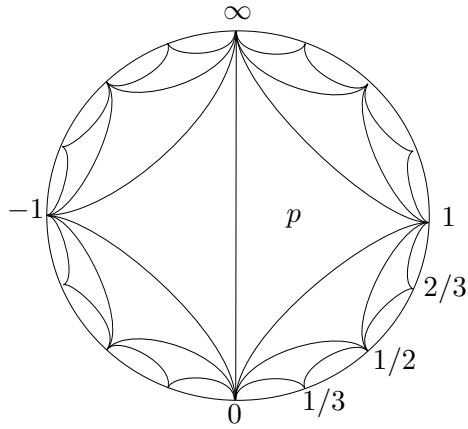%

  \caption{The Farey triangulation}
  \label{fig:farey} 
\end{figure}

Consider two simple closed curves $\gamma_1,\gamma_2$ on the punctured torus $\Sigma$ which intersect once. Then stacking them in the Kauffman algebra $K(\Sigma,A)$ and smoothing the crossing we get $A\delta_1+A^{-1}\delta_2$ where $\delta_1$ and $\delta_2$ are the two curves obtained by smoothing $\gamma_1\cup\gamma_2$. On the Farey triangulation, we simply read that the product of two curves at the boundaries of an edge is a weighted sum of the two curves which are the vertices of the two triangles sharing the original edge.
This immediately shows that the curves associated to the vertices of a fixed triangle generate the Kauffman algebra. 

Consider the curves $\gamma,\delta$ and $\xi$ given in Subsection \ref{punctured-torus}. Suppose that they correspond to $\infty,0$ and $1$ respectively in the Farey triangulation. For any curve $\eta$ such that $T^\eta_r$ is trigonometric, we define its degree as the highest $k\ge 0$ such that $F_{k}^{\eta}$ or $F_{-k}^{\eta}$ is non-zero: it is also equal to the number of non-vanishing diagonals of the matrix of $T^\eta_r$. Then we finally prove Theorem \ref{thmprincipal} in the case of the punctured torus.

\begin{theorem*}
For any simple curve $\eta$ on the once punctured torus, the operator $T_r^{\eta}$ is trigonometric, its degree is equal to the geometric intersection of $\eta$ and $\gamma$, and denoting by $\sigma^\eta$ its $\psi$-symbol, we have
$$\sigma^{\gamma}=f_{\gamma}+\left(\frac{1}{2i}\sum_e\frac{\partial^2f_{\gamma}}{\partial \theta_e\partial \tau_e}\right)\hbar+o(\hbar)$$
\end{theorem*}
\begin{proof}
If $\eta$ corresponds to $\infty,0$ or $1$, it is proved by the direct computations of Subsection \ref{punctured-torus}.
Given any $\eta$ parameterized by a vertex $q=[a,b]$ in the Farey triangulation, let $n$ be number of edges crossed by the geodesic $[p,q]$ where $p$ is a point in the triangle $[0,1,\infty]$: we call it the depth of $\gamma$. We observe that the geometric intersection of $\eta$ with $\gamma$ is equal to $|b|$: we call this quantity the degree of $\eta$. 
Then we show the theorem by induction on the depth. Suppose that it holds for all curves with depth less than $n$ and choose $q$ of depth $n+1$. Then the last edge crossed by the geodesic $[p,q]$ joins two curves $\gamma_1,\gamma_2$ of depth $n$. Let $\gamma_3$ be the symmetric of $\eta$ with respect to the geodesic $[\gamma_1,\gamma_2]$: this vertex has depth $n-1$ and we have $\eta=A^{-1}\gamma_1\cdot\gamma_2-A^{-2}\gamma_3$. As the degree of $\gamma_3$ is strictly less than the sum of the degrees of $\gamma_1$ and $\gamma_2$, we obtain that the degree of $\eta$ is the sum of the degrees of $\gamma_1$ and $\gamma_2$. A direct application of Proposition \ref{produit} shows that the theorem holds for $\eta$ of depth $n+1$ and the theorem is proved.
\end{proof}

Let $\Sigma$ be a sphere with four marked points $p_1,p_2,p_3,p_4$. We can use the same argument as before using the following trick: identify $\Sigma$ with the quotient $\R^2/\Z^2\rtimes\Z_2$ where $\Z^2\rtimes\Z_2$ acts on $\R^2$ by $(n_1,n_2,\epsilon).(x_1,x_2)=(n_1+\epsilon x_1,n_2+\epsilon x_2)$ and we have $p_1=[(0,0)],p_2=[(1/2,0)],p_3=[(0,1/2)]$ and $p_4=[(1/2,1/2)]$.

The double cover $T=\R^2/\Z^2$ is a ramified double cover of $\Sigma$. Any simple curve $\gamma$ in $\Sigma\setminus\{p_1,p_2,p_3,p_4\}$ lifts to a curve in $T$. We associate to it the corresponding slope in the Farey triangulation. If two curves $\gamma_1$ and $\gamma_2$ on $\Sigma$ meet twice, the corresponding vertices in the Farey triangulation form an edge.
Finally, one can write in $K(\Sigma,A)$ the formula $\gamma_1\cdot\gamma_2=A^2\delta_1+A^{-2}\delta_2+R$ where $\delta_1,\delta_2$ are the non-trivial curves obtained by smoothing $\gamma_1\cup\gamma_2$ and $R$ is a combination of boundary curves which can be treated as constants.
Hence, we deduce the following theorem:
 \begin{theorem*}
For any simple curve $\eta$ on the 4 times punctured sphere, the operator $T_r^{\eta}$ is trigonometric, its degree is equal to half the geometric intersection of $\eta$ and $\zeta$, and denoting by $\sigma^\eta$ its $\psi$-symbol, we have
$$\sigma^{\gamma}=f_{\gamma}+\left(\frac{1}{2i}\sum_e\frac{\partial^2f_{\gamma}}{\partial \theta_e\partial \tau_e}\right)\hbar+o(\hbar)$$
\end{theorem*}

\section{Curve operators as Toeplitz operators}\label{toe}

In this section we will establish the fact that the matrices $T^\gamma_r$ defined in the preceding section are the matrices
of Toeplitz operators on the sphere, associated to symbols with specific regularity.
\subsection{Toeplitz operators}
\subsubsection{Matrices of Toeplitz operators}
In this section we will consider the quantization of the sphere in a very down-to-earth way.
Given an integer $N$, we define the space $\mathcal H_N$  of polynomials in the complex variable $z$ of order less than
$N-1$ and set 
\be\label{basis}
\la P,Q\ra=\frac{i}{2\pi}\int_\C\frac{P(z) \ba{Q(z)}}{(1+|z|^2)^{N+1}}dzd\bar z\quad \text{and}\quad\vp^N_n (z)=\sqrt{\frac{N!}{n!(N-1-n)!}}z^n
\ee
The vectors $(\vp^N_n)_{n=0\dots N-1}$ form an orthonormal basis of $\mathcal H_N$.

By the stereographic projection $$S^2\ni (\tau,\theta)\in [0,1]\times S^1\to z=\sqrt{\frac \tau{1-\tau}}e^{i\theta}\in \C\cup\{\infty\},$$ 
The space $\mathcal H_N$ can be seen as a space of functions on the sphere (with a specific behavior at the north pole).
Write $d\mu_N$ the measure $\frac{i}{2\pi}\frac{dzd\bar z}{(1+|z|^2)^{N+1}}$. As the space of analytic functions in $L^2(\C,d\mu_N)$,  the space $\mathcal  H_N$ is closed.

For $z_0\in \C$, we define the coherent state \be\label{repro}\rho_{z_0}(z)=N(1+\bar z_0 z)^{N-1}.\ee These vectors satisfy $\la f,\rho_{z_0}\ra=f(z_0)$ for any $f\in \mathcal H_N$ and the orthogonal projector $\pi_N:L^2(\C,d\mu_N)\to \mathcal H_N$ satisfies $(\pi_N\psi)(z)=\la \psi,\rho_z\ra$.

For $f\in \ci (S^2,\R)$ we define an operator $T_f:\mathcal H_N\to\mathcal H_N$ by the equation $T_f\psi=\pi_N(f\psi)$. A Toeplitz operator on $S^2$ is a sequence of operators $(T_N)\in \mathrm{End} (\mathcal H_N)$ such that 
there exists a sequence $f_k\in\ci(S^2,\R)$ such that for any integer $M$ the operator $R^M_N$ defined by the equation $$T_N=\sum_{k=0}^M N^{-k}T_{f_k}+R_N^M$$ is a bounded operator whose norm satisfies $||R_M||=O(N^{-M-1})$.

An easy use of the stationary phase Lemma shows that the (anti-)Wick symbol of $T_f$, namely $\frac {\la T_f\rho_z,\rho_z\ra}{\la\rho_z,\rho_z\ra}$ satisfies
\be\label{lap}
\frac {\la T_f\rho_z,\rho_z\ra}{\la\rho_z,\rho_z\ra}=f+ \frac 1  N \Delta_S f +O(N^{-2}),
\ee
where $\Delta_S=(1+\vert z\vert^2)^2\partial_z\partial_{\bar z}$ is the Laplacian on the sphere.

\subsubsection{From matrix elements to the total symbol}\label{mateletotal}
In this section we give an explicit formula for computing the exact Toeplitz symbol out of the matrix elements of a given operator expressed by a matrix 
in the basis $\vp_n^N$ defined by \eqref{basis}. The key idea is the following remark.

The matrix element of an operator of the form $T_f:\mathcal H_N\to\mathcal H_N$ between $\vp_n^N$ and $\vp_m^N$ is given by the formula

\begin{eqnarray*}\label{me1}
F_{m,n}&=&\langle T_f\vp_n^N,\vp_m^N\rangle\\
&=&\notag{}\!\!
\frac{i}{2\pi}\int f(z,\bar z)z^n\bar z^m\frac{dzd\bar z}{(1+\vert z\vert^2)^{N+1}}\sqrt{\frac{N!}{n!(N-1-n)!}\frac{N!}{m!(N-1-m)!}}
\end{eqnarray*}
If we suppose that $f$ is real valued, the matrix of $T_f$ is hermitian and setting $z=\sqrt\rho e^{i\theta}$ we can expand $f$ into Fourier series:

$$f(\sqrt\rho e^{i\theta},\sqrt\rho e^{-i\theta})=\sum_{\mu\in\Z}f_\mu(\rho) e^{i\mu\theta}.$$ 
where $f_{-\mu}=\ba{f_\mu}$. Therefore we have that:
\be\label{me2}
\left\{\begin{array}{rcl}
 F_{n,n+\mu}&=&\int_0^\infty f_\mu(\rho)\rho^{n+\frac{\mu}2}\frac{d\rho}{(1+\rho)^{N+1}}\frac{N!}{\sqrt{n!(N-1-n)!((n+\mu)!(N-1-n-\mu)!}}\\
\ba{F_{n,n-\mu}}&=&\int_0^\infty f_\mu(\rho)\rho^{n-\frac{\mu}2}\frac{d\rho}{(1+\rho)^{N+1}}\frac{N!}{\sqrt{n!(N-1-n)!((n-\mu)!(N-1-n+\mu)!}}
\end{array}\right.
\ee

Let us define the Mellin transform $\mathcal M f$ of $f$ as
\be\label{mellin}
\mathcal Mf(s)=\int_0^\infty x^s f(x)\frac{dx}x
\ee
and its inverse
\be\label{mellinv}
\mathcal M^{-1}F(x)=\frac 1{2\pi i}\int_{c-i\infty}^{c+i\infty} x^{-s} F(s)ds.
\ee
We easily get
\begin{proposition}\label{mel1}
\be\label{me3}
f_\mu(\rho)=\frac{1}{2i\pi}\rho^{-1-\frac\mu 2}(1+\rho)^{N+1}
\int_{c-i\infty}^{c+i\infty} F_{s,s+\mu}\rho^{-s}C(N,\mu,s)ds
\ee
and
\be\label{me3bis}
\ba{f_\mu(\rho)}=\frac{1}{2i\pi}\rho^{-1+\frac\mu 2}(1+\rho)^{N+1}
\int_{c-i\infty}^{c+i\infty} F_{s,s-\mu}\rho^{-s}C(N,-\mu,s)ds
\ee
for a convenient value of $c\in \R$ and where 
$$C(N,\mu,s)=\frac{\sqrt{\Gamma(s+1)\Gamma(s+\mu+1)\Gamma(N-s)\Gamma(N-s-\mu)}}{\Gamma(N+1)}$$
\end{proposition}

Write $C(N,\mu,s)=E(N,\mu,s)\Pi(N,s)$ where 
$$\Pi(N,s)=\frac{\Gamma(s+1)\Gamma(N-s)}{\Gamma(N+1)}.$$
We compute that 
\begin{equation}\label{defE}
E(N,\mu,s)=\left\{
\begin{array}{c}
\big(\prod_{k=1}^\mu \frac{s+k}{N-s-k}\big)^{1/2}\text{ if }\mu>0\\
\big(\prod_{k=1}^{-\mu} \frac{N-s-1+k}{s+1-k}\big)^{1/2}\text{ if }\mu<0\\
1\text{ if }\mu=0
\end{array}\right.
\end{equation}
\begin{theorem}\label{totalsymbolabs}
Fix $N$ and $k$ two integers with $k<N$. Let $T$ be an operator in $\mathcal H_N$ whose matrix elements 
are $F_{m,n}$ in the basis $\{\vp_n^N\}_{n=0\dots N-1}$ and vanish for $|m-n|> k$. Suppose that the following assumptions hold:
\begin{itemize}
\item[-] For any integer $\mu$ such that $|\mu|\le k$, there is an analytic extension of $F_{s,s+\mu}E(N,\mu,s)$ to a holomorphic  function $G(\mu,s)$ on the strip $ M_1+k/2<\re(s)+\mu/2<M_2-k/2$.
\item[-] 
For all $s$ and $\mu$, $\ba{G(-\mu,s+\mu)}E(N,\mu,\ba{s})^2=G(\mu,\ba{s})$
\item[-] There are constants $\alpha,\beta$ such that $|G(\mu,s)|\le \alpha \exp(\beta|\im s|)$ for all $s$ such that $ M_1+k/2<\re(s)+\mu/2<M_2-k/2$. Moreover $\beta$ satisfies $\beta<\pi$.
\end{itemize}
Suppose that $M_1<-k-1$ and $M_2>N+k$. Then the formulas of Proposition \ref{mel1} define a function $f$ on the sphere such that $T=T_{f}$. Moreover $f$ is 
\begin{itemize}
\item[-]of class $C^{\infty}$ on the sphere minus the poles,
\item[-]of class $C^{M_S}$ at the south pole  where $M_S<-M_1-1-k$
\item[-]of class $C^{M_N}$ at the north pole where $M_N<M_2-N-k$.
\end{itemize}
\end{theorem}

\begin{remark}
The second hypothesis is an enhancement of the Hermitian condition $F_{\ba{s},\ba{s}+\mu}=\ba{F_{s+\mu,s}}$, observing that for $s\in \R, E(N,-\mu,s+\mu)=E(N,\mu,s)^{-1}$. 
It follows that $G(\mu,s)$ vanishes when $E(N,\mu,s)^2$ does, that is if $s+\mu\in\{0,\ldots,N-1\}$ but $s\notin\{0,\ldots,N-1\}$.
\end{remark}

\begin{proof}

Let us start with the following remark,  also useful in Section \ref{genus2}.
\begin{remark}\label{derder}
In order that a function $f$ on the sphere, regular away from the poles, defines an operator through the formula $T=T_f$ at fixed value of $N$, we only need that all the matrix elements exist, which means that the function has to be integrable at the poles. This property is reflected by the holomorphy conditions on the functions $G(\mu,s)$ as we are going to see now.
\end{remark}

Using many times the formula $\Gamma(s+1)=s\Gamma(s)$ and the reflection formula $\Gamma(s)\Gamma(1-s)=\frac{\pi}{\sin(\pi s)}$ we get 

\be\label{pii}\Pi(N,s)=\frac{\pi s(1-s)(2-s)\cdots(N-1-s)}{\sin(\pi s)N!}\ee
If $\re(s)$ remains bounded and writing $\xi=\im s$, we deduce the following uniform estimate: 
\be\label{estimpi}
\Pi(N,s)\underset{|\xi|\to\infty}{\sim}\frac{\pi |\xi|^N}{e^{\pi|\xi|}N!}
\ee
Rewriting the first formula of Proposition \ref{mel1} we have 
\be\label{mel3}
f_\mu(\rho)=\frac{1}{2i\pi}\rho^{-1-\frac\mu 2}(1+\rho)^{N+1}
\int_{c-i\infty}^{c+i\infty} G(\mu,s)\rho^{-s}\Pi(N,s)ds
\ee

The estimation $|G(\mu,s)|\le \alpha\exp(\beta |s|)$ implies that the integral in the proposition is well defined and smooth for $\rho\in (0,\infty)$ as soon as $\beta<\pi$. In order to get the regularity property we need to show that
\begin{itemize}
\item[-] $f_\mu\rho^{-\frac\mu2}\in C^{M_S}$ for $\rho\sim 0$.
\item[-] $f_\mu\rho^{+\frac\mu2}\in C^{M_N}$ for $\rho\sim \infty$.
\end{itemize}
We do the proof for $\rho\sim 0$, the case $\rho\sim\infty$ being the same.

Taking the constant $c=M_1+\frac{k-\mu}{2}+\epsilon$ we get by the residue formula
$$f_\mu(\rho)=-\rho^{-1-\frac\mu2}(1+\rho)^{N+1}\sum_{c<\ell<0} G(\mu,l)
\binom{N-l-1}{-l-1}
\rho^{-\ell}+O(\rho^{-M_1-1-\frac{k}{2}-\epsilon})$$
If $\mu>0$, we have that $G(\mu,l)=0$ if $l<0$ and $l+\mu\ge 0$. Hence the first non-zero term in the sum is for $l=-1-\mu$. This shows that one can factor $\rho^{\frac \mu 2}$ in the sum as expected. If $\mu<0$, the first non zero residue is for $l=-1$ and this time, one can factor the term $\rho^{-\frac \mu 2}$.

We deduce the $C^{M_S}$ regularity property where $M_S$ is the highest integer with $M_S<-M_1-1-k$.   
\end{proof}

\subsubsection{Asymptotic expansion}

\begin{theorem}\label{asymptotic}
Let $T^N$ be a sequence of operators in $\mathcal H_N$ whose matrix elements are denoted by $F^N_{n,n+\mu}$. Let $g(\mu,\tau,\frac 1 N)$ be a family of functions holomorphic in $\tau$ satisfying 
\be\label{coeffs}
g(\mu,\frac{n}{N},\frac{1}{N})=F_{n,n+\mu}E(N,\mu,n)\textrm{ for all }n\text{ such that }n,n+\mu\in \{0,\ldots,N-1\}
\ee
We suppose that 
\begin{itemize}
\item[-] There exists $M_1<0<1<M_2$ such that $g(\mu,\tau,\frac 1 N)$ are holomorphic on the strip $M_1+\frac{k}{2N}<\re(\tau)+\frac{\mu}{2N}<M_2-\frac{k}{2N}$ and have an asymptotic expansion when $N$ goes to infinity which is uniform on compact sets. 
\item[-] $\ba{g(-\mu,\tau+\frac{\mu}{N},\frac 1 N)}E(N,\mu,N\ba{\tau})^2=g(\mu,\ba{\tau},\frac{1}{N})$ for all $N$ and $\tau$.
\item[-] There are constants $\alpha,\beta$ such that one has 
$$|g(\mu,\tau,\frac 1 N)|\le \alpha e^{\beta|\im \tau|}\text{ for }\tau\text{ in the domain}.$$  
\end{itemize}
Denote by $f^N$ the sequence of functions given by Theorem \ref{totalsymbolabs} so that $T^N=T_{f^N}$. 
Then $f^N$ has an asymptotic expansion of the form 
$$f^N\sim\sum_{n=0}^{\infty}f^{(n)} N^{-n}$$
where $f^{(n)}$ are smooth functions on $S^2$ and $f^{(0)}(\tau,\theta)=\sum_\mu f^{(0)}_\mu(\tau)e^{i\mu\theta}$ where $f_\mu^{(0)}(\tau)=\lim\limits_{N\to \infty}F_{\lfloor N\tau\rfloor ,\lfloor N\tau\rfloor+\mu}$.
 In other terms, the family $T^N$ is a family of Toeplitz operators on the sphere whose principal symbol is $f^{(0)}$. 
\end{theorem}
\begin{proof}

Suppose that $\mu\ge 0$ and consider the second formula of Proposition \ref{mel1}:
$$f_{-\mu}^N(\rho)=\frac{1}{2i\pi}\rho^{-1+\frac\mu 2}(1+\rho)^{N+1}
\int_{c-i\infty}^{c+i\infty} g(-\mu,\frac{s}{N},\frac 1 N)\rho^{-s}\Pi(N,s)ds.$$

Change the variable $s=-1+\al$ with $\al=\frac\rho{1+\rho}+i\xi, \xi\in\R$.

We get 
$$f_{-\mu}^N(\rho)=\frac{1}{2i\pi}(1+\rho)^{N+1}\rho^{\mu/2}
\int g(-\mu,\frac{-1+\al}{N},\frac 1 N) \rho^{-\al}\Pi(N,-1+\al)d\al$$
Let us define $\chi_M\in\ci(\R)$ such that $\chi_M(\xi)=1$ for $|\xi|<M$ and $\chi_M(\xi)=0$ for $|\xi|>M+1$. 
Performing a decomposition of the identity \\$1=\chi_M(\xi/N)+(1-\chi_M(\xi/N))$, 
we write $$f_{-\mu}^N(\rho)=\frac{1}{2i\pi}(1+\rho)^{N+1}\rho^{\mu/2}(J(N,\rho)+I(N,\rho)).$$

We first consider the integral $I(N,\rho)=$
$$\rho^{-\frac\rho{1+\rho}}
\int \!\!g(-\mu,\frac{-1+\frac\rho{1+\rho}+i\xi}{N},\frac 1 N) \rho^{-i\xi}\Pi(N,-1+\frac\rho{1+\rho}+i\xi)(1-\chi_M(\xi/N))d\xi$$

We  show that $|(1+\rho)^{N+1}I(N,\rho)|=O(N^{-\infty})$ uniformly for $\rho$ in compacts subsets of $[0,\infty)$. 
Indeed we have
$$ |I(N,\rho)|\le \!\int\limits_{|\xi|>MN}\!\!C e^{\frac \beta N |\xi|}|\Pi(N,-1+\frac\rho{1+\rho}+i\xi)|d\xi\le C_1\! \int\limits_{|\xi|>MN} \!\!e^{\frac \beta N |\xi|}\frac{|\xi|^N}{e^{\pi |\xi |}N!}d\xi$$
Here we used the estimate for $\Pi(N,-1+i\xi)$ given in Equation \eqref{estimpi}.
We compute explicitly 
$$\int_{|\xi|>MN} e^{\frac \beta N |\xi|}\frac{|\xi|^N}{e^{\pi |\xi |}N!}=\frac{2}{(\pi-\frac\beta N)^{N+1}}\frac {1} {N!} \int_{M(N\pi-\beta)}^\infty u^N e^{-u}du$$

Let $N$ be big enough such that $M(N\pi-\beta)\ge (M+1) N$ and $\pi-\frac \beta N \ge 1$. We get $|I(N,\rho)|\le C_2 (M+1)^Ne^{-MN}$ as an application of the following elementary lemma:
\begin{lemma}
For $a\ge 1$, we have $\frac{1}{N!}\int_{a N}^\infty u^N e^{-u}d u\le a^N e^{(1-a)N}$.
\end{lemma}
This shows that $(1+\rho)^{N+1}I(N,\rho)$ is $O(N^{-\infty})$ uniformly for $\rho\le \frac{e^M}{2M}$ and hence on any compact sets as $M$ can be chosen arbitrarily big.

For $J(N,\rho)$, remarking that 
$\Pi(N,s)=\frac{\Gamma(s+1)\Gamma(N-s)}{\Gamma(N+1)}=\int_0^{\infty}\frac{u^s d u}{(1+u)^{N+1}}$ if $-1<\re s<N$ and multiplying $\al$ by $N$,  we get 

$$J(N,\rho)=N\int g(-\mu,-1/N+\al,\frac 1 N) 
\rho^{-N\al}\int_0^\infty \frac{u^{-1+N\al} d u}{(1+u)^{N+1}}\chi(\Im\al)d\al$$
Putting $u=\rho v$ in the second integral and using Fubini Theorem we get 
\begin{eqnarray*}
J(N,\rho)&=&
N\int d\al\int_0^\infty dv\, g(-\mu,-1/N+\al,\frac 1 N)\chi(\Im\al) \frac{v^{-1+N\al}}{(1+\rho v)^{N+1}}\\
&=&
N\int d\al\int_0^\infty dv\, g(-\mu,-1/N+\al,\frac 1 N) \chi(\Im\al)\frac{e^{N(\al\ln v-\ln(1+\rho v))}}{v^{}(1+\rho v)}\label{oufff}
\end{eqnarray*}

The phase $\Phi(\al,v)=\al\ln v-\ln(1+\rho v)$ is stationary for $v=1$ and $\al=\frac{\rho}{1+\rho}$, that is $v=1,\xi=0$.
Moreover the determinant of the Hessian at the stationary point is equal to 
$1$
; therefore we can apply the stationary phase lemma uniformly up to $\rho\to 0$ for the part on the integral near the stationary point. 
For the rest of the integral we notice that $\Phi(\frac\rho{1+\rho}+i\xi,v)=-\log{(1+\rho)}+\Phi_1(\xi,v)$ with $\Re\Phi_1< 0$ for $\vert\xi\vert >0$. Therefore the non-stationary phase lemma
 gives  that the contribution
to $J(N,\rho)$ given by $\vert\xi\vert >0$ in the integral \eqref{oufff} is $O((1+\rho)^{-N}N^{-\infty})$

We get 
$$J(N,\rho)\sim\frac{2i\pi}N (1+\rho)^{-(N+1)}\sum_{k=0}^\infty N^{-k}g_k(-\mu,\frac{\rho}{1+\rho})$$
where $g_0=g(-\mu,\frac{\rho}{1+\rho},0)$ and $g_k(-\mu,\tau)$ are smooth functions of  $\tau\in[0,+\infty)$.

Putting everything together we get 

$$f^N_{-\mu}(\rho)\sim\rho^{\mu/2}\sum_{k=0}^\infty N^{-k}g_k(-\mu,\frac{\rho}{1+\rho})$$
This proves that for any $k$, the function $f_k(\tau,\theta)=\sum_{\mu}\rho^{-\mu/2}g_k(-\mu,\tau)e^{i\mu\theta}$ is a smooth function at the south pole.

From Equation \eqref{coeffs} we get $G(-\mu,\tau,\frac 1 N)=F_{N\tau,N\tau-\mu}E(N,-\mu,N\tau)$ and from Equation \eqref{defE} we have $E(N,-\mu,N\tau)=(\frac{1-\tau}{\tau})^{\mu/2}=\rho^{-\mu/2}$. 
Hence, $f^N_\mu(\rho)\simeq f^{(0)}_\mu(\tau)$ as claimed in the theorem.

\end{proof}

\subsection{Curve operators as Toeplitz operators}\label{toe2}
Let us relate Theorem \ref{totalsymbolabs} to TQFT and curve operators.
Let $\Sigma$ be either the once punctured torus or the 4 times punctured sphere. Let $r$ be an integer and $\cc$ be an admissible coloring of the marked points of $\Sigma$. We consider the following two cases:
\begin{enumerate}
\item 
If $\Sigma$ is a torus, let $r$ be the level and $a$ the (odd) color of the marked point. The basis of $V_r(\Sigma,a)$ is parametrized by an integer $m\in(a/2,r- a/2)$. Hence its dimension is $N=r-a$.
We identify the basis $\phi_m$ of $V_r(\Sigma,a)$ and the basis $\phi_n^N$ of $\mathcal H_N$ by setting $m=n+\frac{a+1}{2}$ for $0\le n<N$.
\item
 If $\Sigma$ is a sphere, we write $\cc=(a,b,c,d)$. The basis of $V_r(\Sigma,\cc)$ is parametrized by an integer $m$ satisfying $\max(|a-d|,|b-c|)<m<\min(a+d,2r-a-d,b+c,2r-b-c)$ and $m+1=a+d=b+c\mod 2$. Its dimension is $N$ where 
 $$N=\frac{1}{2}\big(\min(a+d,2r-a-d,b+c,2r-b-c)-\max(|a-d|,|b-c|)\Big).$$
 We  index the basis $\phi_m$ of $V_r(\Sigma,\cc)$ and $\phi_n^N$ of $\mathcal H_N$ by setting $m=\max(|a-d|,|b-c|)+1+2n$ for $0\le n<N$.
\end{enumerate}
In any case, we define an isomorphism $I_r:V_r(\Sigma,\cc)\simeq \mathcal{H}_{N}$.
Fix a curve $\gamma$ on $\Sigma$. The curve operator $T^\gamma_r$ gives a matrix $I_rT^\gamma_r I_r^{-1}\in \textrm{End}(\mathcal{H}_{N})$ that we also denote by $T^{\gamma}_r$ for short.
We first prove in Subsection \ref{analytic-continuation} the following theorem:

\begin{theorem}\label{thmtoe2}
Let $\Sigma$ be either the once punctured torus or the 4 times punctured sphere, $r$ a level and $\cc$ an admissible coloring of the marked points. Let $N=\dim V_r(\Sigma,\cc)$ and let $I_r:V_r(\Sigma,\cc)\simeq \mathcal H_N$ be the isomorphism described above. 

Let $\gamma$ be a curve on $\Sigma$ of degree $k$. 
Suppose that the quantity $M_P=\frac{a}{2}-1-k$ is non negative in the torus case and that the quantities 
$M_S=M-1-k$ and $M_N=r-m-k$ are non negative in the sphere case where $M=\frac{1}{2}\max(|a-d|,|c-d|)$ and $m=\frac{1}{2}\min(a+d,2r-a-d,b+c,2r-b-c)$.

Then, there is a function $f^\gamma_r$ on the sphere such that $$T^\gamma_r=T_{f^\gamma_r}.$$
This function $f^{\gamma}_r$ is smooth away from the poles and is 
\begin{itemize}
\item[-] of class $\lfloor M_P\rfloor$ at the poles in the torus case.
\item[-] of class $\lfloor M_S\rfloor$ at the south pole and of class $\lfloor M_N\rfloor$ at the north pole in the case of the 4 times punctured sphere.
\end{itemize}
\end{theorem} 

\subsubsection{Analytic continuation of matrix elements}\label{analytic-continuation}
Let $\gamma$ be a curve of degree $k$ in $\Sigma$. Theorem \ref{thmprincipal} states that there exists coefficients $F_{n,n+\mu}$ for $|\mu|\le k$ such that $T_r^\gamma\phi_n=\sum_{\mu}F_{n,n+\mu}\phi_{n+\mu}$.

The formulas for $F_{n,n+\mu}$ are rational expressions involving square roots, trigonometric functions, the parameter $n$, the coloring $\cc$ and the level $r$. There is an obvious extension of $F_{n,n+\mu}$ to a multivalued analytic function $F_{s,s+\mu}$ whose domain contains $[0,N-1]\cap ([0,N-1]-\mu)$. 
\begin{definition}\label{Gcoeff}
For any curve $\gamma$, set 
\be\label{eq:Gcoeff}
G^{\gamma}(\mu,s)=F_{s,s+\mu}E(N,\mu,s).
\ee
\end{definition}
We will also use the following definition in order to estimate $G(\mu,s)$.
\begin{definition}
Let $f(s)$ be some meromorphic function in a domain of the form $x<\re(s)<y$. Given some $\beta\ge 0$ we will say that $f$ has order $\beta$ if there exists $M,\alpha,\alpha'>0$ such that  for any $s$ with $|\Im(s)|\ge M$ and $x<\re(s)<y$ we have 
$$\alpha e^{\beta|\im(s)|}\le |f(s)|\le \alpha'e^{\beta |\im(s)|}.$$
\end{definition}
\begin{proposition}\label{analyticTQFT}
Let $\Sigma$ be either a once-punctured torus or a 4 times punctured sphere. Let $\cc$ be an admissible coloring of of level $r$. Let $N$ be the dimension of $V_r(\Sigma,\cc)$ and consider the isometric spaces $V_r(\Sigma,\cc)\simeq \mathcal{H}_{N}$.

For any curve $\gamma\subset\Sigma$ of degree $k$, let $G^{\gamma}(\mu,s)$ be the analytic functions associated to $T^\gamma_r$ by Definition \ref{Gcoeff}. Then 
\begin{enumerate}
\item $G^{\gamma}(\mu,s)$ is holomorphic for  $M_1+k/2<\re(s)+\mu/2<M_2-k/2$. 
\item There exists $\beta$ depending only on $\gamma$ such that for all $\mu$, $G(\mu,s)$ have order $\frac{\beta}{r}$.
\end{enumerate}
Where $M_1=-a/2$ and $M_2=N-1+a/2$ in the torus case. In the sphere case, $M_1=-M$ and $M_2=N-1+r-M$ where $M=\frac{1}{2}\max(|a-d|,|b-c|)$.
\end{proposition}
\begin{proof}
The proof consists in checking it for the generators and use Kauffman relations.
The second property is obvious in all the explicit formulas so we will not deal with it.

{\bf Torus case}

Let $\gamma$ be the curve of Figure \ref{fig:torep}. Then $\mu=0$ and $E(N,0,s)=1$. The two statements are obvious from the following expression:
$$G^\gamma_N(0,s)=-2\cos(\frac{\pi}{2r}(2s+a+1)).$$ 

The curve $\delta$ satisfies 
$$G^\delta_N(1,s)=-\sqrt{\frac{\sin(\frac{\pi}{r}(s+a+1))\sin(\frac{\pi}{r}(s+1))}
{\sin(\frac{\pi}{2r}(2s+a+1))\sin(\frac{\pi}{2r}(2s+a+3))}}E(N,1,s).$$
The square of the first factor has poles at $s=-(a+1)/2$ and $s=-(a+3)/2$ modulo $r$ and zeroes at $s=-1$ and $s=r-a-1$ modulo $r$. 
On the other hand, one has $E(N,1,s)=\sqrt{\frac{s+1}{N-s-1}}$. It follows that $G^{\delta}(1,s)$ is holomorphic for $-\frac{a+1}{2}<s<r-\frac{a+3}{2}$ as expected.

The same kind of computation shows that $G^{\delta}_N(-1,s)$ is holomorphic for $-\frac{a-1}{2}<s<r-\frac{a+1}{2}$. Finally, the formula $$G^{\xi}_N(\pm 1,s)=G^{\delta}_N(\pm 1,s)\exp\Big(\frac{i\pi}{2r}(\pm 2s+a+2)\Big)$$ shows that the analytic properties of $G^{\delta}$ and $G^{\xi}$ are the same.

We observe that we can set $M_1=-a/2$ and $M_2=r-a/2-1=N-1+a/2$. It is clear that all the $G$ functions involved here have ordre at most $\frac{\pi}{2r}$.

{\bf Sphere case}\\
It is very similar to the torus case so that we do not give more details. The general case follows from the following lemma:

\begin{lemma}
Let $T^1,T^2$ be two operators in $\mathcal{H}_N$ of respective degrees $k_1,k_2$ whose coefficients can be analytically continued that is, there is a family $G^i(\mu,s)$ for $i=1,2$ and $|\mu|\le k_i$ such that Equation \eqref{eq:Gcoeff} is satisfied and the $G^i(\mu,s)$ satisfy the properties of Theorem \ref{analyticTQFT}.

Then, the product $T^2T^1$ has the same property, meaning that its coefficients can be analytically continued via functions $G^{12}(\mu,s)$ satisfying the properties of Theorem \ref{analyticTQFT} with $k=k_1+k_2$.
\end{lemma}
\begin{proof}
From the matrix multiplication we get for any $n\in\Z$ the following formula: 
\be\label{produit1}
F^{12}_{n,n+\mu}=\sum_{\nu}F^1_{n,n+\nu}F^2_{n+\nu,n+\mu}
\ee
Using the identity $E(N,\mu,s)=E(N,\nu,s)E(N,\mu-\nu,s+\nu)$ we get:
\be\label{produit2}
G^{12}(\mu,s)=\sum_\nu G^1(\nu,s)G^2(\mu-\nu,s+\nu).
\ee
Take $\mu$ with $|\mu|\le k$ and consider one term of the sum with $|\nu|\le k_1$ and $|\mu-\nu|\le k_2$. 

The factor $G^1(\nu,s)$ is holomorphic if $M_1+\frac{k_1}2<\re(s)+\nu/2<M_2-\frac{k_1}2$, while the second factor is holomorphic if $M_1+\frac{k_2}2 <\re(s+\frac{\nu}2)+\frac{\mu-\nu}2<M_2-\frac{k_2}2$. We deduce that if $M_1+\frac k 2<\re(s)+\frac{\mu}2<M_2-\frac k 2$, the product is holomorphic, proving the first point.

It is obvious from Equation \eqref{produit1} that if $G(\nu,s)$ has order $\frac{\beta_1}{r}$ and $G^2(\nu,s)$ has order $\frac{\beta_2}r$ for all $\mu$ and $\nu$, then $G(\mu,s)$ has order $\frac{\beta_1+\beta_2}r$. This proves the last point. 

\end{proof}
The proof of Theorem \ref{analyticTQFT} follows from the particular cases and the product formula by an induction very similar to the one in the proof of Theorem \ref{thmprincipal}.
\end{proof}

\subsubsection{Asymptotic regime}\label{asymptoregime}

Let $\Sigma$ be either the punctured torus or the 4 times punctured sphere. Let $D$ be an odd level and $\cc$ be an admissible coloring of the marked points of $\Sigma$. For any odd integer $\br$, set $r=D\br$. The coloring $\br\cc$ is admissible, hence the family of vector spaces $V_{r}(\Sigma, \br\cc)$ is well-defined and its dimension grows linearly. 
To be more precise, consider the two cases we handle:
\begin{enumerate}
\item 
If $\Sigma$ is a torus, let $D$ be the level and $a$ the color of the marked point. The space $V_r(\Sigma,a\br)$ has dimension $\br\Delta$ where $\Delta=D-a$.
\item
 If $\Sigma$ is a sphere, we write $\cc=(a,b,c,d)$. The space $V_r(\Sigma,\br\cc)$ has dimension $\br\Delta$ where 
 $$\Delta=\frac{1}{2}\big(\min(a+d,2D-a-d,b+c,2D-b-c)-\max(|a-d|,|b-c|)\big)$$
\end{enumerate}
In any case, we define a sequence of isomorphisms $I_r:V_r(\Sigma,\br\cc)\simeq \mathcal{H}_{\br\Delta}$ by identifying the canonical hermitian basis of both spaces. 

\begin{theorem}\label{curvetoplitz}
Let $\Sigma$ be either the once punctured torus or the 4 times punctured surface and $\gamma$ be a curve in $\Sigma$. By Theorem \ref{analyticTQFT}, for any $r=\br D$, let $f^{\gamma}_r$ be the function on the standard sphere such that $T^\gamma_r=T_{f^\gamma_r}$. 

Suppose that in the sphere case one has either $a\ne d$ or $b\ne c$. Then the sequence $f^\gamma_r$ as an asymptotic expansion as 
$$f^\gamma_r=\sum_{n\ge 0}^{\infty} r^{-n} f^\gamma_{(n)}$$ where the functions $f^{\gamma}_{(n)}$ are $\ci$ on the sphere and $f^{\gamma}_{(0)}=f_{\gamma}$ 
In other words, the family $T^\gamma_r$ is a Toeplitz operator with principal symbol $f_\gamma$. 
\end{theorem}
This is an application of Theorem \ref{asymptotic}. The only hypothesis which was not proved 
in Theorem \ref{analyticTQFT} is an analytic version of Theorem \ref{thmprincipal} which is stated below.

Moreover the proof of Theorem \ref{asymptotic} permits to compute, by using sub-leading order in stationary phase Lemma, 
 the subprincipal symbol of the curve operators. However, we will give in Subsection \ref{subleading} a less tedious way of getting it.

%The operator can by expressed as $f(z,\hbar z\frac d {dz})+\left(f(z,\hbar z\frac d {dz})\right)^\star+O(\hbar^2)$ 
%where $f=f_0+\hbar f_1$ 
%given by keeping only the $\mu> 0$ and half of the $\mu=0$ part of \eqref{magic2}, 
%and the ordering chosen as in Section \ref{wkb}. 
%So we can easily compute the begining of the asymptotic expansion of $\frac{\la\rho_z,T_.\rho_z\ra}{\la\rho_z,\rho_z\ra}$, 
%which, by \eqref{lap} is equal to $f_{(0)}+\hbar (f_{(1)}+\frac 1 2 \Delta_S f_{(0)})+O(\hbar^2)$. 
%Using \eqref{magic3} and Lemma \ref{standard} we get easily that
%\be\label{subenfin}
%f_{(1)}(z,\tau)=\Re\left[-\tau\partial_\tau f_{(0)}(z,\tau)-\frac {\tau(1-\tau)} 2 \partial^2_{\tau^2}f_{(0)}(z,\tau)\right]
%\ee 
%which is basically \eqref{sub}.

\begin{remark}
It is important in the hypothesis of Theorem \ref{curvetoplitz} that one has $M_1<0<N-1<M_2$. This will ensure that for any $\gamma$ the exact symbols $f^{\gamma}_r$ will have increasing regularity at the poles when $r$ is large. 
This property is provided by the hypothesis in Theorem \ref{curvetoplitz}. No hypothesis is needed in the torus case as the condition $a>0$ is automatic from the oddness of $a$.

We can relate these conditions to the singularities of the moduli spaces. 
In the torus case, the moduli space $\boM(\Sigma,\frac{\pi a}{D})$ is singular only if $a=0$ (we suppose that $\Delta=D-a>0$ and hence $a\ne D$).
In the sphere case, it is well known that the moduli space $\boM(\Sigma,\frac{\pi a}{D},\frac{\pi b}{D},\frac{\pi c}{D},\frac{\pi d}{D})$ is singular if and only one has $a\pm b\pm c\pm d \ne0\mod 2D$ for all choices of signs. In particular, Theorem \ref{curvetoplitz} applies in all regular cases. 
\end{remark}

\begin{remark}
We remark that we could have checked that the theorem holds for the generators and invoke the well-known but non trivial fact that a product of Toeplitz operators is a Toeplitz operator.
\end{remark}

\begin{proposition}
Let $\Sigma$ be either a once-punctured torus or a 4 times punctured sphere. Let $\cc$ is an admissible coloring of level $D$. Let $\br$ be an odd integer, write $r=D\br$ and $N=\Delta\br$. Let $M_1$ and $M_2$ be defined as in Theorem \ref{analyticTQFT} we suppose that $M_1<0$ and $M_2>N-1$.

For any curve $\gamma\subset\Sigma$ of degree $k$, write $g^{\gamma}(\mu,\tau,\frac 1 N)=F^{N}_{s,s+\mu}E(N,\mu,N\tau)$ where $F^N_{n,m}$ are the matrix elements of $T^\gamma_r$, $\tau=\frac s N$. 
Then, $g^\gamma(\mu,\tau,\frac 1 N)$ is an analytic function for $M_1+\frac{k}{2N}<\re(\tau)+\frac{\mu}{2N}<M_2-\frac{k}{2N}$. On any compact subsets of the strip $M_1<\re(\tau)<M_2$, one has the following expansion:
$$g^\gamma(\mu,\tau,\hbar)=\sum_{n=0}^{\infty}N^{-n}g^{\gamma}_{(n)}(\mu,\tau)$$
Moreover, one has
$$f_\gamma=\sum_\mu \Big(\frac{1-\tau}{\tau}\Big)^{\mu/2}g^{\gamma}_{(0)}(\mu,\tau)e^{i\mu\theta}$$
\end{proposition}

\begin{proof}
The first result is obvious on the generators of the Kauffman algebra and is stable by taking  products. The second statement is a direct consequence of Theorem \ref{thmprincipal}.
\end{proof}

\section{Asymptotics of mapping group representations}\label{mapping}
The aim of this section is to present some applications of the Toeplitz calculus introduced in the preceding section to the computation 
of the asymptotics of the coefficients of the quantum representations of the mapping class group.
In particular we want to recover in a systematic way the results of \cite{tw} on the aymptotics of the  6j-symbols.   
\subsection{General pairing}
In this section we compute the leading order of the scalar product between any two eigenvectors of any two curve-operators in the 
limit $N\to\infty$ under the condition that the intersection of the underlying Bohr-Sommerfeld curves are transversal.

\subsubsection{WKB quasi-modes}\label{wkb}
Since the curve operators are Toeplitz operators on the sphere it is well known that each eigenvector can be approximated by WKB type quasi-modes, analytic in a neighborhood of the underlying Bohr-Sommerfeld level set, see  \cite{pu,voros}. 
We want to present here a computation valid microlocally away from the singular points of the symbol, and better adapted to the algebraic properties of the TQFT than the usual construction. The computation will be valid far away form the poles of the sphere, but since it is known that the construction is analytic in a neighborhood of the BS curve, the formula will be valid all around the trajectory. This method will just happen to be more efficient for the algebra of the computation. Moreover it will emphasize the role of the matrix elements of the curve operator, and will be possibly generalizable to higher genus situations.

Let $\Sigma$ be either the punctured torus or the 4 times punctured sphere and $\cc$ be a coloring of the marked points with level $D$. Given any odd number $\br$, we fix $N=\Delta\br$ and $r=D\ba{r}$ and consider the isomorphic spaces $V_r(\Sigma,\br\cc)\simeq \mathcal H_N$ as in Subsection \ref{asymptoregime}.
We fix $\hbar=\frac 1N$ 
and $\vp^N_n(z)=\sqrt{\frac{N!}{n!(N-1-n)!}}z^n$ as previously. 

Given a curve $\gamma\subset \Sigma$ of degree $k$, we denote by $T^\gamma_r$ the curve operator acting on $\mathcal H_N$ 
and we denote by $F_\mu(\tau,\hbar)$ the functions satisfying 
\be\label{matrixcoeff}
T^\gamma_r \phi^N_n=\sum_{|\mu|\le k}F_\mu(n\hbar,\hbar)\phi^N_{n+\mu}
\ee
 %Set $H=\frac 1Nz\frac d {dz}$ 
 % and $$T^+=(H(1-H))^{-\frac 1 2}z(1-H-\hbar)\mbox{ and }T^-=((H+\hbar)(1-H-\hbar))^{-\frac 1 2}\hbar\frac d {dz}.$$ 
 Define the following operators:
 $$H\vp^N_n=\hbar n\vp^N_n,\ T^+\vp^N_n=(1-\delta_{n,N-1})\vp^N_{n+1}\mbox{ and }T^-\vp^N_n=(1-\delta_{n,0})\vp^N_{n-1}.$$ Therefore $T^\gamma_r=\sum_{|\mu|\le k}(T^{\textrm{Sign}(\mu)})^{\vert\mu\vert}F_\mu(H,\hbar)$.
We have $H=\frac 1Nz\frac d {dz}$. Pick $\epsilon>0$ and consider the space $\mathcal H_N^\epsilon=\textrm{span}\{\phi_n, \epsilon N<n<(1-\epsilon)N\}$. 

We first remark that, on $\mathcal H_N^\epsilon$,
\begin{equation}\label{tplusmoins}
T^\pm=\left(\sqrt{\frac{1-H}H}z\right)^\pm.
\end{equation} 
Since \eqref{tplusmoins} is valid only  on $\mathcal H_N^\epsilon$,
we are going to perform the following computations  in restriction to $\mathcal H_N^\epsilon$. 
In other words, we restrict microlocally the phase space far away from the poles of the sphere. Under this condition we can replace
$T^\pm$ by $\left(\sqrt{\frac{1-H}H}z\right)^\pm$. %This is not necessary, but make the computations much lighter.

%Thanks to the formula
%$$z\vp^n=\sqrt{\frac H{1-H}}\vp^{n+1}$$
We can write on $\mathcal H_N^\epsilon$,
\be
T^\gamma_r=\sum_{|\mu|\le k}\left(\sqrt{\frac {1-H}H}z\right)^{\mu}F_\mu(H,\hbar).
\ee

Since
$$Hz=z(H+\hbar)$$
we have, again on $\mathcal H_N^\epsilon$,
$$\sqrt{\frac {1-H}H}z=z\sqrt{\frac {1-(H+\hbar)}{H+\hbar}}=z\sqrt{\frac {1-H}H}\left(1-\frac \hbar 2\frac1{H(1-H)}\right)+O(\hbar^2).$$

We prove easily the following lemma.
\begin{lemma}
We have the following uniform estimate on $\mathcal H_N^{\epsilon}$: 
\be\left(\sqrt{\frac {1-H}H}z\right)^\mu=z^\mu\left(\sqrt{\frac {1-H}H}\right)^\mu(1-\hbar\frac{\mu(\mu+1)}{4H(1-H)}+
O_{\mathcal B(\mathcal H_N^\epsilon\to\mathcal H_N)}(\hbar^2)).
\ee
\end{lemma}

Let us define $f(z,x,\hbar)=\sum_\mu z^{\mu}G_\mu(x,\hbar)$ where, for $\epsilon<x<1-\epsilon$,
\be\label{gmu}
G_\mu(x,\hbar)=\left(\frac {1-x}x\right)^{\mu/2}F_\mu(x,\hbar)(1-\hbar\frac{\mu(\mu+1)}{4x(1-x)}+O(\hbar^2)).
\ee
We obviously have on $\mathcal H_N^\epsilon$ that
\be\label{fcourbe}
T^\gamma_r=f(z,H,\hbar)+
O_{\mathcal B(\mathcal H_N^\epsilon\to\mathcal H_N)}(\hbar^2).
\ee

The following is standard:
\begin{lemma}\label{standard}
Let $T=\sum_\mu z^\mu G_\mu(H,\hbar)$. Let $W(z)=W_0(z)+\hbar W_\hbar(z)$ such that
\be
f(z,zW_0'(z))=E
\ee 
Then
\be\label{stand}
Te^{\frac {W(z)}\hbar}=
\left(f(z,zW_0'(z))+\frac \hbar2\partial_x^2f(z,zW_0'(z))(z\frac d {dz})^2W(z)+O(\hbar^2)\right)e^{\frac {W(z)}\hbar}.
\ee
\end{lemma}
\begin{proof}
Once again \eqref{stand} is standard in Toeplitz quantization and is essentially ``algebraic". Nevertheless we give a direct proof. 
Remember that the ordering is chosen by
$T=\sum_\mu z^\mu G_\mu(H,\hbar)$. Therefore we have that
$$
T=\int\hat f(z,t)e^{itH}dt\chi_{[0,1]}(H)
$$ where $\hat f$ is the Fourier transform in the second variable of a Schwartz function on the real line equal to $f(z,x)$ 
on $0\leq x\leq 1$, so that $\hat f$ is in the Schwartz class.
 We remark now that, since $H=\hbar z\frac d {dz}$, $e^{itH}\psi(z)=\psi(e^{it\hbar}z)=\psi(z)
+it \hbar z\psi'(z)-\frac{t^2 \hbar^2}2(z^2\psi''(z)+z\psi'(z))+o(\hbar^2)$ from which we deduce the formula \eqref{stand}.
\end{proof}
It is well known by semiclassical Toeplitz theory  (see e.g. \cite{pu,voros}) that, 
for the regular part of the spectrum of a Toeplitz operator, the eigenvectors are,  in a 
neighborhood of the energy shell, close to the WKB expression. Let us recall the argument : one can construct a WKB expression which is analytic and single determined in such a neighborhood. After projection by the Toeplitz projector, this gives rise to a vector in the Hilbert space, close to the WKB expression  in the 
neighborhood.
Using Lemma \ref{standard} we can compute the WKB quasimode up to $\hbar^2$, $e^{\frac {W(z)}\hbar}$, the following way.

Let $f(z,x)=f_0(z,x)+hf_1(z,x)+O(h^2)$ and $W=W_0+\hbar W_1$ with
\be\label{hj}
f_0(z,zW'_0(z))=E,
\ee
Thanks to formula \eqref{stand} we compute the first order correction:
\begin{xalignat*}{1}
\partial_x f_0(z,zW'_0(z))z&W_1'(z)+f_1(z,zW'_0(z))=-\frac 12 \partial_x^2 f_0(z,zW'_0(z))(z\frac d {dz})^2W_0(z).\\
&=-\frac 12 (z\frac d {dz})(\partial_xf_0(z,zW'_0(z)))
+\frac 12 z\partial_z\partial_x f_0(z,zW'_0(z)).
\end{xalignat*}

Therefore:
$$
W_1(z)=-\frac 12\log{(\partial_xf_0(z,zW'_0(z)))}+\int\phi(z) dz$$
where $
\phi(z)=\frac{\frac 12 z\partial_z\partial_x f_0(z,zW'_0(z))-f_1(z,zW'_0(z))}
{z\partial_x f_0(z,zW'_0(z))}$.
According to Theorem \ref{thmprincipal}, we have:
\be\label{formulemagique}
F_\mu(x,\hbar)=F_\mu(x,0)+\hbar\frac {\mu+1} 2\partial_x F(x,0)+O(h^2),
\ee
Notice that $\mu$ has changed to $\mu+1$ because the change between the variable $\tau=\frac{\pi m}{r}$ and the variable $x=\frac{\pi n}{r}$ has the form $\tau=\alpha+\beta(x+\frac{\hbar}{2})$ for some constants $\alpha$ and $\beta$. The shift by $\frac \hbar 2$ is the important Maslov correction.

We compute $f_0(z,x)=\sum_\mu z^\mu\big(\frac{1-x}{x}\big)^{\mu/2}F_\mu(x,0)$ and using Formula \eqref{formulemagique} we get:
\be\label{magic2}f_1(z,x)=\sum_\mu z^\mu\big(\frac{1-x}{x}\big)^{\mu/2}\big(\frac {\mu+1} 2\partial_x F(x,0)-\frac{\mu(\mu+1)}{4x(1-x)}F_\mu(x,0)\big).
\ee

%Let us remark that
%\be\label{magic3}
%f_1(z,\frac{z\bar{z}}{1+z\bar{z}})=\frac 1 2 \Delta_S f_0(z,\frac{z\bar{z}}{1+z\bar{z}})
%\ee
%where $\Delta_S=(1+\vert z\vert^2)^2\partial_z\partial_{\bar z}$ is the Laplacian on the sphere.

Equation \eqref{magic2} implies that $\frac 12 z\partial_z\partial_xf_0(z,x)-f_1(z,x)=-\frac 12 \partial_x f_0$ and hence $\phi(z)=-\frac{1}{2z}$. Hence we obtain:
\be\label{quasiquasi}
e^{\frac{W(z)}h}=\frac 1{\sqrt{z\partial_x f_0(z,zW'_0(z))}}e^{\frac{W_0(z)}h}.
\ee

Let $H(z,\ba{z})=f_0(z,\frac{|z|^2}{1+|z|^2})$. This  is the classical hamiltonian in the sense that $H(z,\ba{z})=\sum_{\mu} F_\mu(\tau,0)e^{i\mu\theta}$ where $(\tau,\theta)$ are the spherical coordinates associated to $(z,\ba{z})$. 

Note that, although the computation has been done away from the two poles, we know that the quasi-mode is holomorphic and single-valued in a neighborhood of the energy shell.
Therefore the formula \eqref{quasiquasi} is valid also near the two poles.

We have $\partial_{\ba{z}}H=\frac{z}{(1+|z|^2)^2}\partial_x f_0$ hence we have proven the following:
\begin{proposition}\label{quasimode}
Let $W_0$ be a holomorphic solution of the Hamilton-Jacobi equation $f_0(z,zW'_0(z))=E$, $E$ regular value 
of the function $f_0(z,\frac{|z|^2}{1+|z|^2})$ satisfying Bohr-Sommerfeld condition. 
Then
$$
\psi_{\hbar}(z)=\frac{(1-zW_0')}{\sqrt{\partial_{\ba z}H(z,\frac{W'_0}{1-zW'_0})}}e^{\frac{W_0(z)}{\hbar}}
$$
is a quasimode modulo $\hbar^2$ of $T_r^\gamma$ where $f$ and $T_r^\gamma$ are related by Equation \eqref{fcourbe}.
\end{proposition}

For any $z$ such that $H(z,\ba z)=E$ we have 
\be 
\psi_\hbar(z)=\frac{1}{1+|z|^2}\frac{e^{W_0/\hbar}}{\sqrt{\partial_{\ba z} H(z,\ba z)}}
\ee

Moreover we have
\be\label{potsymp}
d\Big(W_0(z)-\frac{1}{2}\log(1+|z|^2)\Big)=\frac{1}{2}\frac{\ba{z}d z-zd\ba{z}}{1+|z|^2}=2i\pi\eta
\ee
 where $d\eta=\omega$ is the symplectic form. 
 Let $z_0$ be a point satisfying $H(z_0,\ba{z_0})=E$. 
 We will say that the quasi-mode $\psi$ is normalized at $z_0$ if $W_0(z_0)=\frac{1}{2}\log(1+|z_0|^2)$. By equation \eqref{potsymp}, this implies that $W_0-\frac{1}{2}\log(1+|z|^2)$ is purely imaginary on the level set.

We derive easily, using the stationary phase lemma, the following lemma.
\begin{lemma}\label{normaliz}
Let $E$ be a regular level of the hamiltonian $H$ with period $T$ and let $\psi$ be the normalized quasi-mode given in Proposition \ref{quasimode}. Then 
$$||\psi_\hbar||^2=\frac{T}{2}\sqrt{2\pi\hbar}+O(\hbar^{-1/2})$$
\end{lemma}

\subsubsection{Pairing formula}\label{pairing12} 
In this section we apply the explicit formulas of the preceding subsection to compute the desired scalar products.
\begin{theorem}\label{pairingformula}
Let $\Sigma$ be either the once punctured torus or the 4 times punctured sphere, $\cc$ an admissible coloring of level $D$ and $\br$ an odd integer. Let $N=\br \Delta$ be the dimension of $V_r(\Sigma,\br\cc)$ where $r=D\br$.

Let $\gamma_0$ and $\gamma_1$ be two curves on $\Sigma$. 
Denote by $T^{\gamma_0}_N$ and $T^{\gamma_1}_N$ the corresponding curve operators and by $H_0$ and $H_1$ the corresponding principal symbols (i.e. minus the trace functions). 

Let $m_0$ and $m_1$ be two natural numbers and set for $i=0,1$: $\Sigma_i=\{z, H_i(z,\ba{z})=-2\cos(\frac{\pi m_i}{r})\}$.

Suppose that $\Sigma_0$ and $\Sigma_1$ are non-empty regular curves which intersect transversally. Pick $z_0\in \Sigma_0$ and $z_1\in\Sigma_1$.
 
Let $\psi_i\in \mathcal H_N$ be a unit eigenvector of $T^{\gamma_i}_N$ with eigenvalue $-2\cos(\frac{\pi m_i}{r})$ and whose phase at $z_i$ is the same as the phase of $\big(\partial_{\ba z}H(z_i,\ba z_i )\big)^{-1/2}$ (notice that this condition defines $\psi_i$ up to a sign). Then, 
\be\label{wood}
\langle\psi_0,\psi_1\rangle=\frac{\pm2\hbar^{1/2}}{\sqrt{T_0T_1}}\sum_{z\in \Sigma_0\cap\Sigma_1}\frac {e^{-i\frac{\pi}{4}\textrm{Sign}(\{H_0,H_1\}(z))}}{\sqrt{|\{H_0,H_1\}(z)|}}e^{\frac{2i\pi}{\hbar}\big(\int_{C_0}\eta-\int_{C_1}\eta\big)}+O(\hbar^{\frac32}),
\ee
where 
\begin{itemize}
\item[-] $\{.,.\}$ is the Poisson bracket on the sphere and $\eta$ is the symplectic potential given in Equation \eqref{potsymp}.
\item[-] For $i=0,1$, $C_i$ is a path in $\Sigma_i$ joining $z_i$ to $z$.
\end{itemize}
\end{theorem}

\begin{remark}
We can always assume that one of the two curve operators is the diagonal one. In this case all the eigenvectors, including 
the one corresponding to the extrema of the spectrum are given \textit{exactly} by the WKB quasimode expression at leading order. 
Therefore the expression
\eqref{wood} is still valid in the case where $\Sigma_0$ is reduced to a single point. 
Note that in this case $\{H_0,H_1\}$ will be of order $\hbar^{\frac 1 2}$ in (and only in) the critical case where 
$dH_0\to 0$ at the the two poles.
\end{remark}

\begin{proof}
Let $T^{\gamma_i}_r,  H_i, E_i,T_i,\psi_i,W_i,a_i$ be respectively the curve operator, the hamiltonian, the level, the period, the quasimode, the phase function and the amplitude associated to each curve. 
Thanks to the normalization condition, the eigenvectors are closed to the WKB quasimodes $\psi_i=a_i e^{W_i/\hbar}$ as $\hbar$ goes to 0, the proof comes from an estimation of the scalar product 

$$\langle \psi_0,\psi_1\rangle=\frac{i}{2\pi}\int 
a_0(z)\ba{a_1(z)}
e^{
\frac{W_0(z)+\ba{W_1(z)}}
{\hbar}
} 
\frac{dz d\ba{z}}{(1+|z|^2)^{\frac{1}{\hbar}+1}}$$

From the Hamilton-Jacobi equation \eqref{hj}, the critical points of $F(z,\ba{z})=W_0(z)+\ba{W_1(z)}-\log(1+|z|^2)$ are precisely the intersection points of the curves $\Sigma_0$ and $\Sigma_1$. Moreover for any $z\in \Sigma_0\cap \Sigma_1$ we compute 

$$\partial^2_z F=-\frac{\partial_zH_0}{\partial_{\ba z}H_0}(1+|z|^2)^{-2}
,\,\partial^2_{\ba z} F=-\frac{\partial_{\ba z}H_1}{\partial_{z}H_1}(1+|z|^2)^{-2},\,\partial_z\partial_{\ba z}F=-(1+|z|^2)^{-2}$$

By the stationary phase lemma, we get that the contribution to the pairing of a neighborhood of $z$ is given by 
$$\frac{a_0(z)\ba{a_1(z)}}{\pi(1+|z|^2)}2\pi\hbar \det(A_1)^{-1/2}e^{\frac{1}{\hbar}\left(W_0+\ba{W_1}-\log(1+|z|^2)\right)}$$
where 
$$A_{\lambda}=(1+|z|^2)^{-2}\begin{pmatrix} 2+\lambda(\frac{\partial_z H_0}{\partial_{\ba z}H_0}+\frac{\partial_{\ba z} H_1}{\partial_{z}H_1})& i\lambda( \frac{\partial_z H_0}{\partial_{\ba z}H_0}-\frac{\partial_{\ba z} H_1}{\partial_{z}H_1})\\
 i\lambda( \frac{\partial_z H_0}{\partial_{\ba z}H_0}-\frac{\partial_{\ba z} H_1}{\partial_{z}H_1})& 2-\lambda(\frac{\partial_z H_0}{\partial_{\ba z}H_0}+\frac{\partial_{\ba z} H_1}{\partial_{z}H_1})\end{pmatrix}$$
Here $\det(A)^{-1/2}$ is defined for $\re(A)\ge 0$  by analytic continuation from symmetric real and positive $A$. We compute $\det(A_0)^{-1/2}=\frac{1}{2}(1+|z|^2)^{2}$ and hence 
$$\det(A_{\lambda})^{-1/2}=\frac{(1+|z|^2)^2}{2\sqrt{1-\lambda^2\frac{\partial_z H_0}{\partial_{\ba z}H_0}\frac{\partial_{\ba z} H_1}{\partial_{z}H_1}}^*}$$
where the square root with the asterisk means the one with positive real part.

Using the formula $\{H_0,H_1\}=\frac{2\pi}{i}(1+|z|^2)^2(\partial_z H_0\partial_{\ba{z}} H_1-\partial_z H_1\partial_{\ba{z}} H_0)$ and the fact that $a_i(z)=\frac{1}{1+|z|^2}\frac{1}{\sqrt{\partial_{\ba z}H_i}}$, we get:

$$\langle \psi_0,\psi_1\rangle=\sqrt{2\pi}\hbar\sum_{z\in \Sigma_0\cap\Sigma_1}
\frac{1}{\sqrt{\partial_{\ba z}H_0 \partial_z H_1}}\sqrt{\frac{\partial_{\ba z}H_0\partial_zH_1}{i\{H_0,H_1\}}}^*e^{\frac{\Phi(z)}{\hbar}}+O(\hbar^2)$$
where we derive from Equation \eqref{potsymp} the following formula:
$$\Phi(z)=W_0(z)+\ba{W_1}(z)-\log(1+|z|^2)=\frac{2i\pi}{\hbar}\big(\int_{C_0}\eta-\int_{C_1}\eta\big)$$
We obtain the formula by dividing by the norm of $\psi_0$ and $\psi_1$.
\end{proof}

\subsection{Application to 6j-symbols and punctured S-matrix}\label{pairing}

\subsubsection{6j-symbols}\label{6j}
Consider the case where $\Sigma$ is a sphere with 4 punctures. Let $\gamma_0$ and $\gamma_1$ be respectively the curves $\zeta$ and $\eta$ shown in Figure \ref{fig:spherep}.
Let $\psi_0$ and $\psi_1$ be normalized eigenvectors of $T^{\gamma_0}_r$ and $T^{\gamma_1}_r$ with eigenvalues $-2\cos(\frac{\pi m_0}{r})$ and $-2\cos(\frac{\pi m_1}{r})$ as in Theorem \ref{pairingformula}. Then 
$$\langle \psi_0,\psi_1\rangle=u
\frac{\sqrt{\sin(\frac{\pi m_0}{r})\sin(\frac{\pi m_1}{r})}}{\sin(\frac{\pi}{r})}\Big\{\begin{array}{ccc} \frac{a-1}{2}&\frac{b-1}{2}&\frac{m_0-1}{2}\\ \frac{d-1}{2}&\frac{c-1}{2}&\frac{m_1-1}{2}\end{array}\Big\}$$
where $\Big\{\begin{array}{ccc} a&b&c\\ d&e&f\end{array}\Big\}$ is the 6j-symbol normalized as in \cite{tw} and $u$ is some complex number of modulus 1.

Recall that the moduli space $\boM(\Sigma,\frac{\pi a}{r},\frac{\pi b}{r},\frac{\pi c}{r},\frac{\pi d}{r})$ is diffeomorphic to the space of spherical quadrilaterals $(P_1,P_2,P_3,P_4)$ whose lengths are $l_{12}=\frac{\pi a}{r},l_{23}=\frac{\pi b}{r},l_{34}=\frac{\pi c}{r},l_{14}=\frac{\pi d}{r}$. 
The symbols of the operators $T^{\gamma_0}_r $ and $T^{\gamma_1}_r$ are $H_0=-2\cos(l_{13})$ and $H_1=-2\cos(l_{24})$. Moreover, the angle coordinates associated to the actions $l_{13}$ and $l_{24}$ are the exterior dihedral edges $\theta_{13}$ and $\theta_{24}$. In particular, one can write $\omega=\frac{r}{4\pi^2N}d\theta_{13}\wedge d l_{13}$  (the normalization factor comes from the constraint $\int\omega=1$).
We deduce from this formula that the period $T_i$ is equal to $\frac{r}{4\pi N\sin(\frac{\pi m_i}{r})}$.
Let $G=\det(\cos(l_{ij}))_{i,j=1}^{4}$. The following formula is proven in \cite{tw}, Proposition 2.4.1: 
$$\frac{\partial l_{ij}}{\partial \theta_{kl}}=-\frac{G^{1/2}}{\sin(l_{ij})\sin(l_{kl})}$$

We deduce from it that $\{H_0,H_1\}=\frac{16\pi^2N}{r}G^{1/2}$. Moreover, following \cite{tw}, the symplectic area $S$ of the region enclosed by $\Sigma_0$ and $\Sigma_1$ equals $\frac r {2\pi^2N} \big( \sum_{a<b} l_{ab}\theta_{ab} -2V\big)$  where $V$ is the volume of the spherical tetrahedron with length $l_{ij}$.

Let $z_0$ and $z_1$ be the two intersection points of the curves $\Sigma_0$ and $\Sigma_1$. Then, $\int_{C_0}\eta-\int_{C_1}\eta=S$.
We deduce that 
$$\langle \psi_0,\psi_1\rangle=u\sqrt{\frac{4}{r}\sin(\frac{\pi m_0}{r})\sin(\frac{\pi m_1}{r})}
\frac{\cos(\frac{S}{2}+\frac{\pi}{4})}{G^{1/4}}+O(r^{-3/2})$$
This gives an alternative proof of the following formula: 
\begin{proposition}\cite{tw}
$$\Big\{\begin{array}{ccc} \frac{a-1}{2}&\frac{b-1}{2}&\frac{e-1}{2}\\ \frac{d-1}{2}&\frac{c-1}{2}&\frac{f-1}{2}\end{array}\Big\}=\frac{2\pi\cos(\frac{ir}{2\pi}\big(\sum_{a<b}l_{ab}\theta_{ab}-V\big)+\frac{\pi}{4})}{r^{3/2}G^{1/4}}+O(r^{-5/2})
$$
\end{proposition}

\subsubsection{Punctured S-matrix}\label{punk}

Let $\Sigma$ be the punctured torus, $\gamma$ and $\delta$ the curves shown in Figure \ref{fig:torep}.
Let $\theta_0\in\R/2\pi\Z$ be the angle coordinate on $\boM(\Sigma,\alpha)$ associated to $\tau_0=\arccos(\frac{1}{2}f_{\gamma})$. By the constraint $\int\omega=1$ we get $\omega=\frac{1}{2\pi(\pi-\alpha)}d\tau_0\wedge d\theta_0$ and from Lemma \ref{toregeom}, we get $f_{\delta}=-2\frac{\sqrt{\sin(\tau_0+\alpha/2)\sin(\tau_0-\alpha/2)}}{\sin(\tau_0)}\cos(\theta_0)=-2\cos(\tau_1)$.

Define $H_0=f_\gamma$, $H_1=f_{\delta}$ and $T_i$ the period of the Hamiltonian flow of $H_i$.
We compute $T_i=\big((\pi-\alpha)2\sin(\tau_i)\big)^{-1}$ and $\{H_0,H_1\}=8\pi(\pi-\alpha)\sin(\theta_0)\sqrt{\sin(\tau_0+\frac\alpha 2)\sin(\tau_0-\frac\alpha 2)}$.

Consider the intersection of the level sets $H_0=-2\cos(\tau_0)$ and $H_1=-2\cos(\tau_1)$. We have the symmetric formula:
$$\{H_0,H_1\}=8\pi(\pi-\alpha)\sqrt{\cos(\frac \alpha 2)^2-\cos(\tau_0)^2-\cos(\tau_1)^2+\cos(\tau_0)^2\cos(\tau_1)^2}.$$
and hence we deduce from Theorem \ref{pairingformula} the following proposition:
\begin{proposition}\label{prop-smat}
Let $\Sigma$ be a punctured torus, $\gamma_0$ and $\gamma_1$ be respectively the curves $\gamma$ and $\delta$ of Figure \ref{fig:torep}.
Let $\psi_0$ and $\psi_1$ be normalized eigenvectors of $T^{\gamma_0}_r$ and $T^{\gamma_1}_r$ with eigenvalues $-2\cos(\tau_0)$ and $-2\cos(\tau_1)$ as in Theorem \ref{pairingformula} where $\tau_i=\frac{\pi m_i}{r}$. Set 
$$G=\cos(\frac \alpha 2)^2-\cos(\tau_0)^2-\cos(\tau_1)^2+\cos(\tau_0)^2\cos(\tau_1)^2.$$
Then we have the following asymptotic formula:
$$\langle\psi_0,\psi_1\rangle=\sqrt{\frac {8\sin(\tau_0)\sin(\tau_1)} {r}}\frac{\cos(\frac{r}{2\pi}\int_{\mathcal D}d\theta\wedge d\tau+\frac\pi 4)}{G^{1/4}}+O(r^{-3/2})$$

In this formula $\mathcal D$ is the domain of $\boM(\Sigma,\alpha)$ defined by the equations $f_{\gamma_0}\ge-2\cos(\tau_0)$ and $f_{\gamma_1}\ge -2\cos(\tau_1)$.
\end{proposition}

Let us deduce Result \ref{smatrix} from this proposition. Let $(\Gamma,c)$ be the colored graph given in Figure \ref{fig:smat} where $c=(m_0,m_1,a)$. Then from TQFT axioms, denoting by $\Gamma_i$ the colored graphs shown in Figure \ref{fig:torep} we have:

$$\langle \Gamma,c\rangle =\frac{||\Gamma_0||||\Gamma_1||}{\eta}\langle \psi_0,\psi_1\rangle$$
where the first bracket stands for the Kauffman bracket and $\eta=\sqrt{\frac 2 r}\sin(\frac \pi r)$ is the quantum invariant of $S^3$. From Formula \eqref{norm}, we get 
\scriptsize
$$\frac{||\Gamma_0||||\Gamma_1||}{\eta}\sim\frac{r^{3/2}}{\sqrt{2}\pi}\left(\frac{\langle m_0+\frac{a-1}{2}\rangle!\langle m_0-\frac{a+1}{2}\rangle!\langle m_1+\frac{a-1}{2}\rangle!\langle m_1-\frac{a+1}{2}\rangle!\langle\frac{a-1}{2}\rangle!^4}
{\langle m_0\rangle!\langle m_0-1\rangle!\langle m_1\rangle!\langle m_1-1\rangle!\langle a-1\rangle!^2}\right)^{1/2}$$
\normalsize
which proves Result \ref{smatrix}, remarking that $\langle m_i\rangle=\sin(\tau_i)$.

\subsection{Wick symbol and the sub-principal symbol}\label{subleading}
The aim of this section is to prove Formula \ref{sub} from the asymptotic of the Wick symbol of the curve operators. Let $\sigma_0,\sigma_1$ be the two first terms in the expansion of the total Toeplitz symbols of the curve operator $T^\gamma_N$. From Equation \eqref{lap}, we get 
$$W(z)=\frac{\la T^\gamma_N \rho_z,\rho_z\ra}{\la \rho_z,\rho_z\ra}=\sigma_0+\frac{1}{N}(\sigma_1+\Delta_S\sigma_0)+O(N^{-2})$$
On the other hand, we have $T^\gamma_N=f_0(z,H)+\frac{1}{N}f_1(z,H)+O(N^{-2})$. 
By a computation similar to Lemma \ref{standard}, we get for a smooth function $g$ the expansion
\begin{eqnarray*}
\frac{g(H)\rho_{z_0}}{\rho_{z_0}}&=&g(\frac{z\ba{z_0}}{1+z\ba{z_0}})+\frac{1}{N}\Big(
-\frac{z\ba{z_0}}{1+z\ba{z_0}}g'(\frac{z\ba{z_0}}{1+z\ba{z_0}})\\
&&+\frac{z\ba{z_0}}{2(1+z\ba{z_0})^2}g''(\frac{z\ba{z_0}}{1+z\ba{z_0}})
\Big)+O(N^{-2})
\end{eqnarray*}

This gives that the Wick symbol of $T^\gamma_N$ is up to $O(N^{-2})$ the same as the one of the Toeplitz operator with symbol
$$f_0(z_0,x_0)+\frac{1}{N}\Big(f_1(z_0,x_0)-x_0\partial_x f_0(z_0,x_0)
+\frac{1}{2}\frac{|z_0|^2}{(1+|z_0|^2)^2}\partial_x^2 f_0(z_0,x_0)
\Big)$$

where we have set $x_0=\frac{|z_0|^2}{1+|z_0|^2}$.
We have $\sigma_0(z_0)=f_0(z_0,x_0)$ and from Equation \eqref{magic2}, we get 
$$\frac{1}{2}\Delta_S \sigma_0(z_0)=f_1(z_0,x_0)-x_0\partial_x f_0(z_0,x_0)+\frac{1}{2}\frac{|z_0|^2}{(1+|z_0|^2)^2}\partial_x^2f_0(z_0,x_0)$$
which implies that $\sigma_1=\frac{1}{2}\Delta_S \sigma_0$ and proves Equation \eqref{sub}.

\section{The genus 2 case}\label{genus2}
 \subsection{The Hilbert space}
Consider the case where $\Sigma$ has genus 2 and no marked points. Pick a pants decomposition such that the graph $\Gamma$ associated to it is a theta graph with edges $e_1,e_2,e_3$ as in Figure \ref{fig:genus2}.
Then, set 
$$U=\{(\tau_1,\tau_2,\tau_3)\in [0,\pi]^3, \,\forall i,j,k\,\tau_i\le \tau_j+\tau_k\textrm{ and }\tau_1+\tau_2+\tau_3\le 2\pi\}.$$
As in Subsection \ref{representation}, the map $p:\boM(\Sigma)\to U$ is an integrable system which has a natural section $s$. We can compare this integrable system with the following well-known system on $\P^3$.

Denote by $Z=[Z_0,Z_1,Z_2,Z_3]$ the homogeneous coordinates on $\P^3$ and set $h_1(Z)=\frac{|Z_2|^2+|Z_3|^2}{|Z|^2}, h_2(Z)=\frac{|Z_1|^2+|Z_3|^2}{|Z|^2},h_3(Z)=\frac{|Z_1|^2+|Z_2|^2}{|Z|^2}$. The map $p=(\pi h_1,\pi h_2,\pi h_3)$ defines an integrable system $\P^3\to U$ with a preferred section $s$ consisting in choosing real and positives values for the $Z_i$s. 

Using the preferred sections and the angle coordinates, we can construct a continuous map $\Phi:\boM(\Sigma)\to \P^3$ such that it commutes with the projections $p$ and the sections $s$, and it is a symplectomorphism over the pre-image of the interior of $U$. We observe that the two spaces are nevertheless distinct as the pre-image of the point $(0,0,0)$ in $\boM(\Sigma)$ is a 3-dimensional variety whereas its pre-image in $\P^3$ is just a point.

A the quantum level, it is natural to expect that $\Phi$ induces an isomorphism between $V_r(\Sigma)$ and the geometric quantization of $\P^3$ at some level. We explicit this isomorphism in the following lines.
Let $\boL\to\P^3$ be the canonical bundle of $\P^3$. We endow it with its hermitian structure and connection such that the curvature is $\frac{1}{i}\omega$ where $\omega=\frac{i}{2\pi}\partial\overline{\partial}\log ||Z||^2$.
The space $H^0(\P^3,\boL^r)$ of holomorphic sections of $\boL^r$ is canonically isomorphic to the space $\C[Z_0,Z_1,Z_2,Z_3]_r$ of homogeneous polynomials of degree $r$.

The scalar product of two sections $\phi,\psi\in H^0(\P^3,\boL^r)$ is defined by $\langle \phi,\psi\rangle=\int_{\P^3}\langle \phi(x),\psi(x)\rangle\mathrm{d}\mu(x)$ where $\mu=\frac{1}{3!}\omega^{\wedge 3}$. At the level of polynomials, this formula makes the monomials orthogonal to each other and 
$$|| \prod_{i=0}^3 Z_i^{n_i}||^2=\frac{\prod_i n_i!}{(\sum n_i+3)!}$$

Let $\C^3\subset \P^3$ be the affine chart defined by $Z_0=1$ and denote by $z_1,z_2,z_3$ the corresponding affine coordinates. Denote by $t$ the coordinate $z_0$ viewed as holomorphic section of $\boL$. Any section of $\boL^r$ has over $\C^3$ the form $f t^r$ for some holomorphic $f$. We compute that in coordinates 
\begin{eqnarray*}
|t|^2&=&\frac{1}{1+|z|^2}\\
\omega&=&\frac{i}{2\pi}\frac{(\sum \mathrm{d}z_i\wedge\mathrm{d}\overline{z}_i)(1+|z|^2)-(\sum \overline{z}_i\mathrm{d}z_i)(\sum z_i\mathrm{d}\overline{z}_i)}{(1+|z|^2)^2}\\
\mu&=&\frac{1}{(2\pi)^3}\frac{\bigwedge_i i\mathrm{d}z_i\wedge\mathrm{d}\overline{z}_i}{(1+|z|^2)^4}\\
||f t^r||^2&=&\frac{1}{(2\pi)^3}\int_{\C^3}\frac{|f(z_1,z_2,z_3)|^2}{(1+|z|^2)^{r+4}}\bigwedge_i i\mathrm{d}z_i\wedge\mathrm{d}\overline{z}_i
\end{eqnarray*}
Hence, we can identify the space $H^0(\P^3,\boL^r)$ with the space of holomorphic functions from $\C^3$ to $\C$ such that $||f t^r||^2<\infty$. 

Let $A_r=\{(\al_1,\al_2,\al_3)\in \mathbb{N}, \al_1+\al_2+\al_3\le r-2\}$ 
this set parametrizes a hermitian basis for $H^0(\P^3,\boL^{r-2})$ where we set 
\be\label{basis2}
\phi^r_{\al_1,\al_2,\al_3}=D(r;\al_1,\al_2,\al_3)z_1^{\al_1}z_2^{\al_2}z_3^{\al_3}
\ee
and $$D(r;\al_1,\al_2,\al_3)=\left( \frac{(r+1)!}{\al_1!\al_2!\al_3!(r-2-\al_1-\al_2-\al_3)!}\right)^{1/2}$$

Recall now that the set of admissible colorings of the theta graph shown in Figure \ref{fig:genus2} is the set 
$$I_r=\{(c_1,c_2,c_3)\in \mathcal{C}_r, \forall i,j,k, c_i< c_j+c_k\text{ and }c_i+c_j+c_k<2r\text{ and odd}\}$$
There is a natural bijection between $A_r$ and $I_r$ given by putting for all $i,j,k$ $c_i=\al_j+\al_k+1$ or equivalently $\alpha_i=\frac{c_j+c_k-c_i-1}{2}$.

Using this bijection, we can identify the spaces $V_r(\Sigma)$ and $H^0(\P^3,\boL^{r-2})$ as Hermitian vector spaces.

Let us also remark that, for any distinct $i,j,k\in\{1,2,3\}$, we have:
\be\label{voila}
\phi^r_{\al_1,\al_2,\al_3}(z_1,z_2,z_3)=\vp^{r+1}_{\alpha_i}(z_i)\vp^{r-\alpha_i}_{\alpha_j}(z_j)
\vp^{r-1-\alpha_i-\alpha_j}_{\alpha_k}(z_k),
\ee
where $\vp^r_\alpha$ is given by \eqref{basis}, out of which we easily derive the expression for the reproducing kernel 
$\rho_{z'_1z'_2z'_3}$:
\be\label{repgen2}
\rho_{z'_1z'_2z'_3}(z_1,z_2,z_3)=(r-1)r(r+1)(1+\overline{z'_1}z_1+\overline{z'_2}z_2+\overline{z'_3}z_3)^{r-2}
\ee
\subsection{Toeplitz operators and matrix elements}

Let us first derive a result partially similar to Proposition \ref{mel1} and Theorem \ref{totalsymbolabs}:
\begin{proposition}\label{mel2}
Let $T$ be an Hermitian operator on $H^0(\P^3,\boL^{r-2})$ whose matrix elements on the basis $\{\phi^r_{\al_1,\al_2,\al_3}, (\al_1,\al_2,\al_3)\in A_r\}$
are denoted by $F_{\al,\al'}$.

Let us suppose that for any $\mu$, there is an ordering $i,j,k\in\{1,2,3\}$  and an analytic extension of 
\be\label{choice}
F_{\al,\al+\mu}E(r,\mu_i,\al_i)E(r-1-\al_i,\mu_j,\al_j)E(r-2-\al_i-\al_j,\mu_k,\al_k)
\ee 
to an holomorphic  function 
$G(\mu,\al)$ 
on a ``strip" $\Omega$ defined by
$$ M<\re(\al_1),M<\re(\al_2),M<\re(\al_3),\re(\al_1+\al_2+\al_3)<M'$$
where $M<0<r-2<M'$ and suppose moreover that, uniformly on $\Omega$ we have 
$$
|G(\mu,s)|\le a \exp(b\sum|\im \al_i|),
$$
for some constants $a$ and $b<\pi$. Then the formulas (write $\rho=\sum_i\rho_i$)
\be\label{mellin3}
f_\mu(\rho_1,\rho_2,\rho_3)=\frac{\prod_i\rho_i^{-1-\frac{\mu_i} 2}}{(2i\pi)^3}(1+\rho)^{r-4}
\int
F_{\al,\al+\mu}(\Pi\rho_i^{-\al_i})C(r,\mu,\al)d\al,
\ee
where $C(r,\mu,\alpha)=D(r;\al)D(r;\alpha+\mu)$ and the integral is over a product of three vertical lines in $\Omega$
\be\label{polar}
f=\sum_{\mu\in\Z^3}f_\mu(\rho_1,\rho_2,\rho_3)e^{i\mu.\theta},\ \theta\in (\R/2\pi\Z)^3,
\ee

define a smooth function on the dense open set $V=\{[z_0,z_1,z_2,z_3]\in \P^3\text{ such that }z_0z_1z_2z_3\ne 0\}$. 
This function satisfies the equation:
\be\label{topgenus2}
T=\mathcal T_f.
\ee
Moreover, writing $\al=\sum_i\al_i$, one has for any $\epsilon>0$ and for $\epsilon<\frac{\al_i}{r}$ and $\al<1-\epsilon$:

$$
F_{\al,\al+\mu}=
f_\mu\left(\frac {\al_1}{r-\al},\frac {\al_2}{r-\al},
\frac {\al_3}{r-\al}\right)+
O(r^{-1}).
$$
\end{proposition}
The proof is an easy but tedious adaptation of the proof of Theorem \ref{totalsymbolabs}. The situation here ``factorizes" 
thanks to 
the property \eqref{voila}, which induces the hypothesis \eqref{choice}. The examples below show that 
different choices of orderings in \eqref{choice} and a version of \eqref{mellin3} similar to \eqref{me3bis}  will be needed.
As shown in Remark \ref{derder}, the existence of the Toeplitz operator requires only some integrability properties (and no regularity
) of the function defined by \eqref{me3}, \eqref{me3bis} so the bounds on $M,\ M'$ are sufficient in this case.
\begin{remark}
Although we will not do it here let us mention that we can derive the regularity properties 
of $f$ at the boundary of $V$ out of holomorphy properties of
the functions $G(\mu,.)$.

The interesting trace functions in the genus-2 case being singular, it is useless to derive a general formula for
the asymptotics of $f$ in the entire $\P^3$. 

\end{remark}

\subsection{Examples}\label{examples}
The first example is an extension of one of the generator for the 4 times punctured sphere.
Let us denote by $\gamma,\delta,\eta$ the curves shown in Figure \ref{fig:genus2}.

\begin{figure}[htbp]
\centering
  \def\svgwidth{8cm}
 \executeiffilenewer{genus2.svg}{genus2.pdf}%
 {inkscape -z -D --file=genus2.svg %
 --export-pdf=genus2.pdf --export-latex}%
 \input{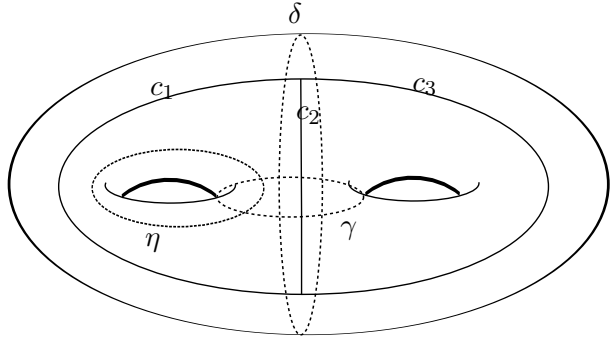}%

  \caption{Curves on a genus 2 surface}
  \label{fig:genus2} 
\end{figure}
One has immediately that the curve operator $T^\gamma_r$ is diagonal and one has $F_{\al,\al}^{\gamma}=-2\cos(\pi\frac{\al_1+\al_3}{r})$. 
 The matrix elements of the curve operator $T^\delta_r$ associated to $\delta$ can be obtained from Proposition \ref{spherep}. We first remark that the only non zero
 ones are given by $\mu=0$ and $\mu=\pm\nu$ where $\nu=(1,-1,1)$. We get
 \begin{eqnarray*}\label{matgenus20}
 F^\delta_{\al,\al}&=&
 -2\cos\frac{\pi(2\al_1+2\al_2+1)}r\\
 &&-4\frac{\sin^2\frac{\pi}r\al_2\sin^2\frac{\pi}r(\al_3+1)}{\sin\frac{\pi}r(\al_1+\al_3+1)\sin\frac{\pi}r(\al_1+\al_3+2)}\\
 &&-4\frac{\sin^2\frac{\pi}r(\al_1+\al_2+\al_3+1)\sin^2\frac{\pi}r(\al_1)}{\sin\frac{\pi}r(\al_1+\al_3)\sin\frac{\pi}r(\al_1+\al_3+1)},
 \end{eqnarray*}
 
 \be\label{matgenus21}
 F^\delta_{\al,\al+\nu}=4\left(
 \frac{\langle\al_1+\al_2+\al_3+2\rangle\langle\al_3\rangle\langle\al_2\rangle\langle\al_1\rangle\langle\al_1+1\rangle}{\langle\al_1+\al_3+2\rangle\langle\al_1+\al_3+3\rangle\langle\al_1+\al_3+2\rangle\langle\al_1+\al_3+3\rangle}\right)^{1/2}
\ee
 and a similar expression for $F^\delta_{\al,\al-\nu}$.

Our second example deals with the first non-reducible case, that is the curve $\eta$ shown in Figure \ref{fig:genus2}. 
A standard computation with fusion rules gives the following formulas where $\mu=(0,0,1)$ and $\nu=(1,-1,0)$. 
\begin{eqnarray*}
F_{\alpha,\alpha+\mu}&=&\frac{\langle\al_1+\al_2+\al_3+2\rangle\langle\al_3+1\rangle}{\big(\langle\al_1+\al_3+1\rangle\langle\al_1+\al_3+2\rangle\langle\al_2+\al_3+1\rangle\langle\al_2+\al_3+2\rangle\big)^{1/2}}\\
F_{\alpha,\alpha-\mu}&=&\frac{\langle\al_1+\al_2+\al_3+1\rangle\langle\al_3\rangle}{\big(\langle\al_1+\al_3\rangle\langle\al_1+\al_3+1\rangle\langle\al_2+\al_3\rangle\langle\al_2+\al_3+1\rangle\big)^{1/2}}\\
F_{\alpha,\alpha+\nu}&=&-\frac{\langle\al_2\rangle\langle\al_1+1\rangle}{\big(\langle\al_1+\al_3+1\rangle\langle\al_1+\al_3+2\rangle\langle\al_2+\al_3\rangle\langle\al_2+\al_3+1\rangle\big)^{1/2}}\\
F_{\alpha,\alpha-\nu}&=&-\frac{\langle\al_1\rangle\langle\al_2+1\rangle}{\big(\langle\al_1+\al_3\rangle\langle\al_1+\al_3+1\rangle\langle\al_1+\al_2+1\rangle\langle\al_1+\al_2+2\rangle\big)^{1/2}}
\end{eqnarray*}

\begin{theorem}\label{symbolgenus}
The operators $T^\gamma_r,\ T^\delta_r,\ T^\eta_r$ are Toeplitz operators with exact symbols $f^\gamma_r,\ f^\delta_r,\ f^\eta_r$ 
defined on $\P^3$, that is
\be\label{tooo}T^\xi_r=\Pi_rf^\xi_r\Pi_r,\quad  \xi\in\{\gamma,\delta,\eta\},\ee
where $\Pi_r$ is the Toeplitz projector associated to $\P^3$.

Moreover $f^\gamma_r,\ f^\delta_r,\ f^\eta_r$ admit an  asymptotic expansion in powers of $\frac 1r$  
smooth in $V$ with leading orders, writing $\tau_i=\frac{\rho_i}{1+\rho}$, 
\begin{eqnarray*}
f^\gamma_0(\tau,\theta)&=&-2\cos(\pi\frac{\tau_1+\tau_3}{r})\\
f^\delta_0(\tau,\theta)&=&2\cos{(\tau_1+\tau_2)}+4\frac{\sin^2{\tau_2}\sin^2{\tau_3}+\sin^2{\tau_1}\sin^2{(\tau_1+\tau_2+\tau_3)}}{\sin^2{(\tau_1+\tau_3)}}\\
&&-8\frac{\sin\tau_1\sin \tau_2\sin \tau_3\sin{(\tau_1+\tau_2+\tau_3)}}{\sin^2{(\tau_1+\tau_3)}}\cos{(\theta_1-\theta_2+\theta_3)}\\
f^\eta_0(\tau,\theta)&=&2\frac{\sin(\tau_1+\tau_2+\tau_3)\sin(\tau_3)}{\sin(\tau_1+\tau_3)\sin(\tau_2+\tau_3)}\cos(\theta_3)\\
&&-2\frac{\sin(\tau_1)\sin(\tau_2)}{\sin(\tau_1+\tau_3)\sin(\tau_2+\tau_3)}\cos(\theta_1-\theta_2).
\end{eqnarray*}
Finally $f^\gamma_0=-\tr\rho(\gamma)$ and $f^\delta_0=-\tr\rho(\delta)$. 
\end{theorem}
\begin{proof}
 The proof consists in checking that the hypothesis of Proposition \ref{mel2} are satisfied for  precise choices of orderings in \eqref{choice}.
This is easily done for most of the terms in $T^\gamma_r$ and $T^\eta_r$, with similar arguments to the ones in section \ref{toe}. 
Still there are some pathological terms which are the ones containing denominators vanishing for extreme values of $\al_1,\al_2,\al_3$.
These terms, though well defined thanks to the numerators, might be non holomorphic. One can check that they still provide the integrability condition necessary for the existence of the Toeplitz operator.
For example the term  $\frac{\sin^2\frac{\pi}r(\al_1)}{ \sin\frac{\pi}r(\al_1+\al_3)}$ in 
$F^\delta_{\al,\al}$  is not holomorphic for negative values of $\al_1$ and $\al_3$. 
Nevertheless we first remark that $\partial^2_{\al_1}F^\delta_{\al,\al}$ and $\partial^2_{\al_3}F^\delta_{\al,\al}$ are integrable near the origin.
Therefore we get by integrating $\al_1,\al_3$ on the pure imaginary axis and $\Re\al_2=-\epsilon<0$, that, after two integrations by part,  
$\vert f^\delta_r \vert\leq C\rho 2^\epsilon(\log(\rho_1)^{-2}(\log(\rho_3)^{-2}(\rho_1\rho_2\rho_3)^{-1}$ which is integrable at the origin 
(see Remark \ref{derder}).
The same pathology appears  with the term
$\frac{\sin^2\frac{\pi}r(\al_3+1)}{ \sin\frac{\pi}r(\al_1+\al_3+2)}$ near $\al_1+\al_3-r-2$ and can be solved the same way.

Let us remark that this discussion shows that the condition of holomorphy of $G(\mu,\al)$  in Proposition 
\ref{mel2} can be weakened by restricting $\Omega$ to $M=0$ and $M'=r-2$ and adding some regularity conditions at the boundary.
\end{proof}
Note that the absence of square roots in the expressions of $f^\delta_0$ and $f^\eta_0$ shows clearly that both are 
singular on the divisor $\{z_0z_1z_2z_3=0\}$.

We leave for a future work the identification of $f^\eta_0$ with the trace function $f_\eta$ on $\boM(\Sigma)$ as suggested by Conjecture \ref{principal}. 

\appendix
\newpage
\section{Computations with fusion rules}\label{fusion}

Here is a toolkit to obtain the formula of Propositions \ref{torep} and \ref{spherep}. 
There is a calculus for colored trivalent graphs in 3-space invented by Kauffman.  We collect here the formulas which are necessary for our computations where we set $[n]=\frac{A^{2n}-A^{-2n}}{A^2-A^{-2}}$. Here an edge colored by $n$ has to be interpreted as $n$ parallel copies of the same edge cabled by the Jones-Wenzl idempotent $f_n$. Be careful that the corresponding colored graph in TQFT has a shifted coloring and an alternate sign, see Subsection \ref{jones-wenzl}. We refer to \cite{mv} for the proofs.

\begin{figure}[htbp]
\centering
  \def\svgwidth{\columnwidth}
 \executeiffilenewer{panneau-idempotent.svg}{panneau-idempotent.pdf}%
 {inkscape -z -D --file=panneau-idempotent.svg %
 --export-pdf=panneau-idempotent.pdf --export-latex}%
 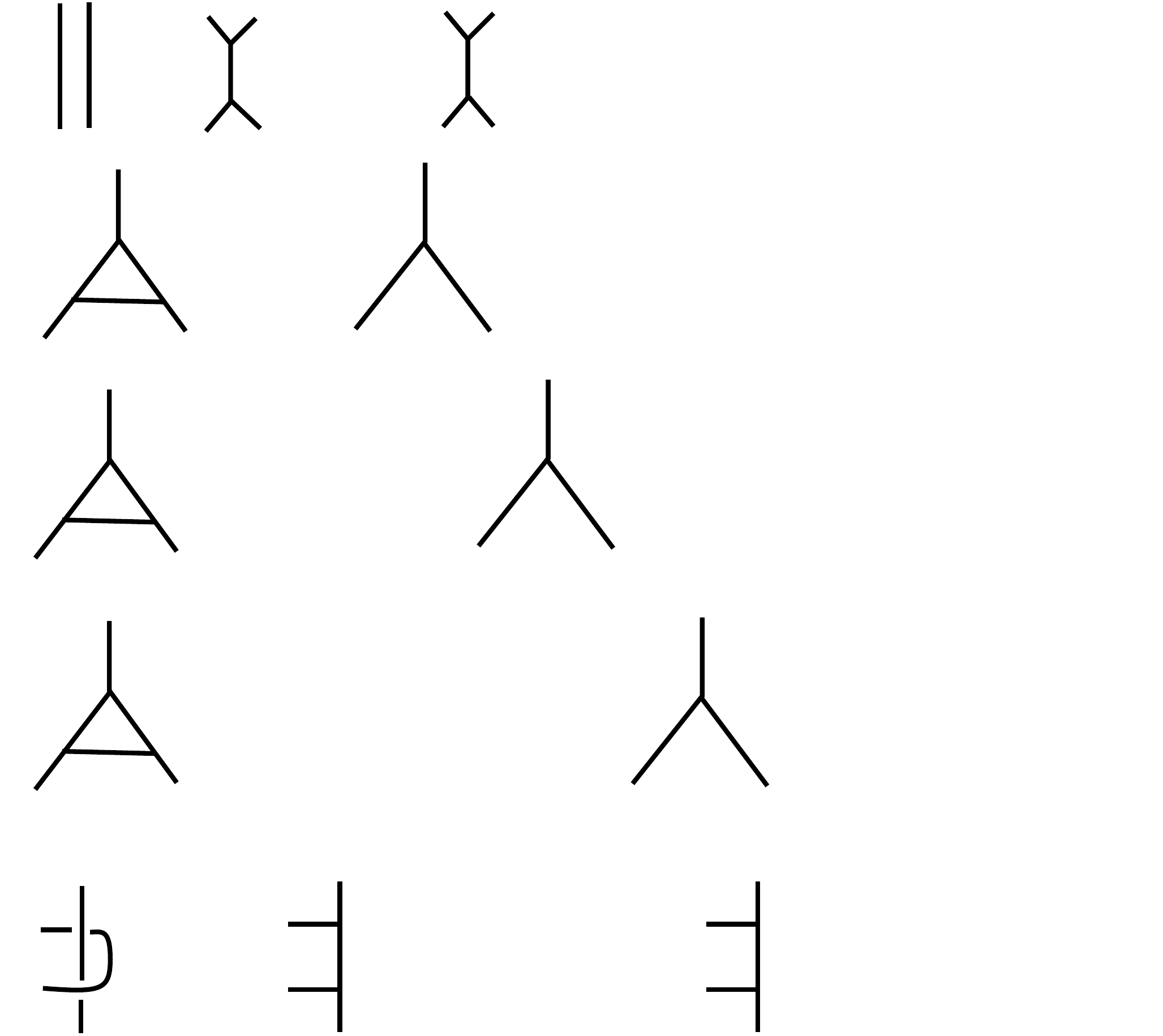%

  \caption{Fusion rules}
  \label{fig:panneau} 
\end{figure}

\end{document}